\documentclass[prx,a4paper,aps,twocolumn,superscriptaddress,10pt]{revtex4-1}

\usepackage{amsmath,amsthm,graphicx,xcolor,times,xfrac,booktabs,mathtools,enumitem,xr,subfigure,bbm,verbatim,appendix,placeins}
\usepackage[unicode=true,bookmarks=true,bookmarksnumbered=false,bookmarksopen=false,breaklinks=false,pdfborder={0 0 1}, backref=false,colorlinks=true]{hyperref}
\renewcommand*{\url}[1]{\href{#1}{#1}}

\makeatletter
\theoremstyle{plain}
\newtheorem{thm}{\protect\theoremname}
\theoremstyle{plain}
\newtheorem{lem}[thm]{\protect\lemmaname}
\theoremstyle{plain}
\newtheorem{prop}[thm]{\protect\propositionname}
\theoremstyle{remark}
\newtheorem*{rem*}{\protect\remarkname}
\theoremstyle{plain}

\theoremstyle{plain}

\theoremstyle{definition}
\newtheorem{defn}[thm]{\protect\definitionname}
\theoremstyle{plain}
\newtheorem*{thm*}{\protect\theoremname}
\theoremstyle{plain}
\newtheorem*{lem*}{\protect\lemmaname}
 
\providecommand{\propositionname}{Proposition}
\providecommand{\theoremname}{Theorem}
\providecommand{\lemmaname}{Lemma}
\providecommand{\remarkname}{Remark}
\providecommand{\conjecturename}{Conjecture}
\providecommand{\definitionname}{Definition}
\providecommand{\corollaryname}{Corollary}
\allowdisplaybreaks

\def\bra#1{\langle{#1}\vert}
\def\ket#1{\vert{#1}\rangle}

\def\BraVert{\egroup\,\mid\,\bgroup}

\def\tr#1{\mbox{tr}\left[{#1}\right]}
\newcommand{\ptr}[2]{\mbox{tr}_{#1}\left[ #2 \right]}
\newcommand{\inp}{\texttt{i}}
\newcommand{\out}{\texttt{o}}

\begin{document}


\title{Quantum Markov Order}
\date{\today}

\author{Philip Taranto}
\email{philip.taranto@monash.edu}
\affiliation{School of Physics \& Astronomy, Monash University, Clayton, Victoria 3800, Australia}

\author{Felix A. Pollock}
\affiliation{School of Physics \& Astronomy, Monash University, Clayton, Victoria 3800, Australia}

\author{Simon Milz}
\affiliation{School of Physics \& Astronomy, Monash University, Clayton, Victoria 3800, Australia}

\author{Marco Tomamichel}
\affiliation{Centre for Quantum Software and Information, School  of Software, University  of  Technology Sydney, Sydney NSW 2007, Australia}

\author{Kavan Modi}
\affiliation{School of Physics \& Astronomy, Monash University, Clayton, Victoria 3800, Australia}


\begin{abstract}
We formally extend the notion of Markov order to open quantum processes by accounting for the instruments used to probe the system of interest at different times. Our description recovers the classical property in the appropriate limit: when the stochastic process is classical and the instruments are non-invasive, \emph{i.e.}, restricted to orthogonal, projective measurements. We then prove that there do not exist non-Markovian quantum processes that have finite Markov order with respect to all possible instruments; the same process exhibits distinct memory effects when probed by different instruments. This naturally leads to a relaxed definition of quantum Markov order with respect to specified sequences of instruments. The memory effects captured by different choices of instruments vary dramatically, providing a rich landscape for future exploration.
\end{abstract}
\maketitle


\textit{Introduction}.---Fundamentally, physical laws are local in time, yet memory effects pervade processes studied throughout the sciences, since no system is isolated~\cite{StochProc}. Our inability to capture interactions between a system of interest and its environment leads to stochastic dynamics for the system, with information about its history influencing future evolution, often leading to a build-up of correlations over time~\cite{BreuerPetruccione}. Such temporal correlations are exhibited over various timescales in complex phenomena. However, a natural notion of \emph{memory length} emerges in the context of statistical modeling: the amount of a system's history that directly affects its future. This, importantly, dictates the resources required for simulation, which grow exponentially in the memory length (even classically)~\cite{Salzberg1998,Thijs2001,Rosvall2014}. Fortunately, most processes have an effectively finite-length memory, permitting an efficient description that considers only the portion of history necessary to predict the future~\cite{Pollock2018T}. Alternatively, given control over some quantum degrees of freedom for some duration, any process with the corresponding memory length can be simulated. Indeed, manipulating memory effects has proven advantageous in various information-processing tasks such as preserving coherence~\cite{NielsenChuang,Viola1999,Banaszek2004,Erez2008,Reich2015}. Clearly, memory will need to be exploited to develop near-term quantum technologies. 

In the classical setting, the finite-length memory approximation underpins the often-invoked order-$\ell$ Markov models, which use information of only the past $\ell$ observed states to predict the next. However, even in the simplest case of \emph{memoryless}, or \emph{Markovian}, dynamics (\emph{i.e.}, $\ell=1$), the study of stochastic processes is vastly different in quantum mechanics than its classical counterpart, mainly because, in the former, one must necessarily disturb the system to observe realizations of the process, breaking an implicit classical assumption~\cite{Modi2011,Modi2012,Modi2012A,Milz2017}. Crucially, this leads to a breakdown of the Kolmogorov extension theorem~\cite{Kolmogorov,Feller,Milz2017KET}, which provides the mathematical foundation of stochastic processes that allows for calculation of conditional probability distributions. Conventional approaches to quantum stochastic processes attempt to sidestep this problem by describing properties of the process in terms of the time-evolving system density operator, failing to capture multi-time effects~\cite{Breuer2016}; others constrain system-environment interactions to specify memory mechanisms~\cite{Giovannetti2012,Lorenzo2017A,Lorenzo2017}; both perspectives lead to necessary but insufficient criteria for Markovianity~\cite{Li2018}. 

The aforementioned issues can be circumvented by separating the controllable influence on the system from the underlying process, as achieved by various modern frameworks, including the process matrix~\cite{Oreshkov2012} and process tensor formalisms~\cite{Pollock2018A,Pollock2018L}. These represent processes as quantum combs~\cite{Chiribella2008,Chiribella2008-2,Chiribella2009}, mapping sequences of probing instruments to accessible joint probability distributions through a generalized spatio-temporal Born rule~\cite{ShrapnelCosta2017}. They have been used to extend the causal modeling paradigm (originally developed for classical processes~\cite{Pearl}) to quantum theory~\cite{Araujo2015,Costa2016, Oreshkov2016, Allen2017, Ringbauer2017}. Most importantly for our purposes, by capturing all multi-time correlations, these frameworks provide unambiguous conditions for a process to be Markovian, unifying all previous approaches~\cite{Pollock2018L}. Like the joint probability distribution characterizing classical stochastic processes, a quantum stochastic process suffers exponentially increasing complexity with respect to its memory length, with the added complication that all possible sequences of interventions must be accounted for. This naturally begs the question: are there quantum processes with finite-length memory, and hence significantly reduced complexity?

In this Letter, we extend the notion of Markov order to quantum stochastic processes. We begin by discussing classical Markov order to motivate its generalization to the quantum realm. We use the process tensor formalism to prove our main result: non-Markovian quantum processes, generically, have infinite Markov order. Afterwards, we formulate the conditions for a quantum process to have finite Markov order in a constrained setting. The structure of quantum Markov order is far richer than its classical counterpart, as explored in detail in an accompanying Article~\cite{Taranto2018A}.


\textit{Classical Markov Order}.---The concept of Markov order is essentialized by the following question: is knowledge of a portion of the history of a process sufficient to predict future statistics? Consider an $(n+1)$--step classical stochastic process, segmented into three intervals: the future $F=\{t_n, \hdots, t_k\}$, the memory $M=\{t_{k-1}, \hdots, t_{k-\ell}\}$, and the history $H = \{t_{k-\ell-1}, \hdots, t_0\}$ (in principle, the history and future can extend infinitely long). The random variables $X_j$ describing the system (with the subscript denoting the timestep) are grouped similarly: $\{X_F, X_M, X_H\}$. The Markov order of the process is defined in terms of the conditional statistics of these random variables:

\begin{defn}\label{def:cmarkovorder} \textbf{(Classical Markov Order)} A classical stochastic process has Markov order-$\ell$ if the conditional probability for any realization $x_F$ of the random variables $X_F$ beyond any time $0< t_k < t_n$ depends only on the realizations $x_M$ of those in the previous $\ell$ timesteps, and not on realizations $x_H$ of the earlier history:
    \begin{align}\label{eq:cmarkovorder} 
        \mathbbm{P}_{F}(x_{F}|x_{M}, x_H) =\mathbbm{P}_{F}(x_{F}| x_{M}). 
    \end{align}
As a special case, $\ell=1$ corresponds to a Markovian process.
\end{defn}

The property of Markov order-$\ell$ constrains the underlying joint probability distribution characterizing the process, from which these conditional distributions arise. It follows from Eq.~(\ref{eq:cmarkovorder}) that, for any realization of events in any length-$\ell$ block $M$, the joint conditional distribution over $F$ and $H$ factorizes: 
\begin{align}\label{eq:cmarkovcondindep} 
    \mathbbm{P}_{FH}(x_F,x_H|x_M) = \mathbbm{P}_F(x_F|x_M) \mathbbm{P}_H(x_H|x_M).
\end{align}
In other words, the future and the history are \emph{conditionally independent}, given specification of events in the memory. This conditional independence is equivalently expressed by the vanishing classical conditional mutual information (\textbf{CMI}), $I(F:H|M)=0$.

While Markov order-$\ell$ dictates that the next state depends only upon the previous $\ell$, it does not imply an absolute demarcation of timesteps into blocks of memory and irrelevant history. Instead, the memory blocks corresponding to different timesteps overlap, permitting the existence of unconditional correlations between timesteps with separation greater than $\ell$ in general~\footnote{These are themselves often referred to as memory. Here, we distinguish the total temporal correlations between observables from those resulting from \textit{non-Markovian} memory, which can survive interventions that reset the system's state.}; however, such correlations are always mediated through overlapping memory blocks (see Fig.~\ref{fig:memory}). Markov order thus quantifies how much of the history one must remember to predict the future, providing a natural ``measure'' for the memory length of the process.


\begin{figure}[t]
\centering
\includegraphics[width=\linewidth]{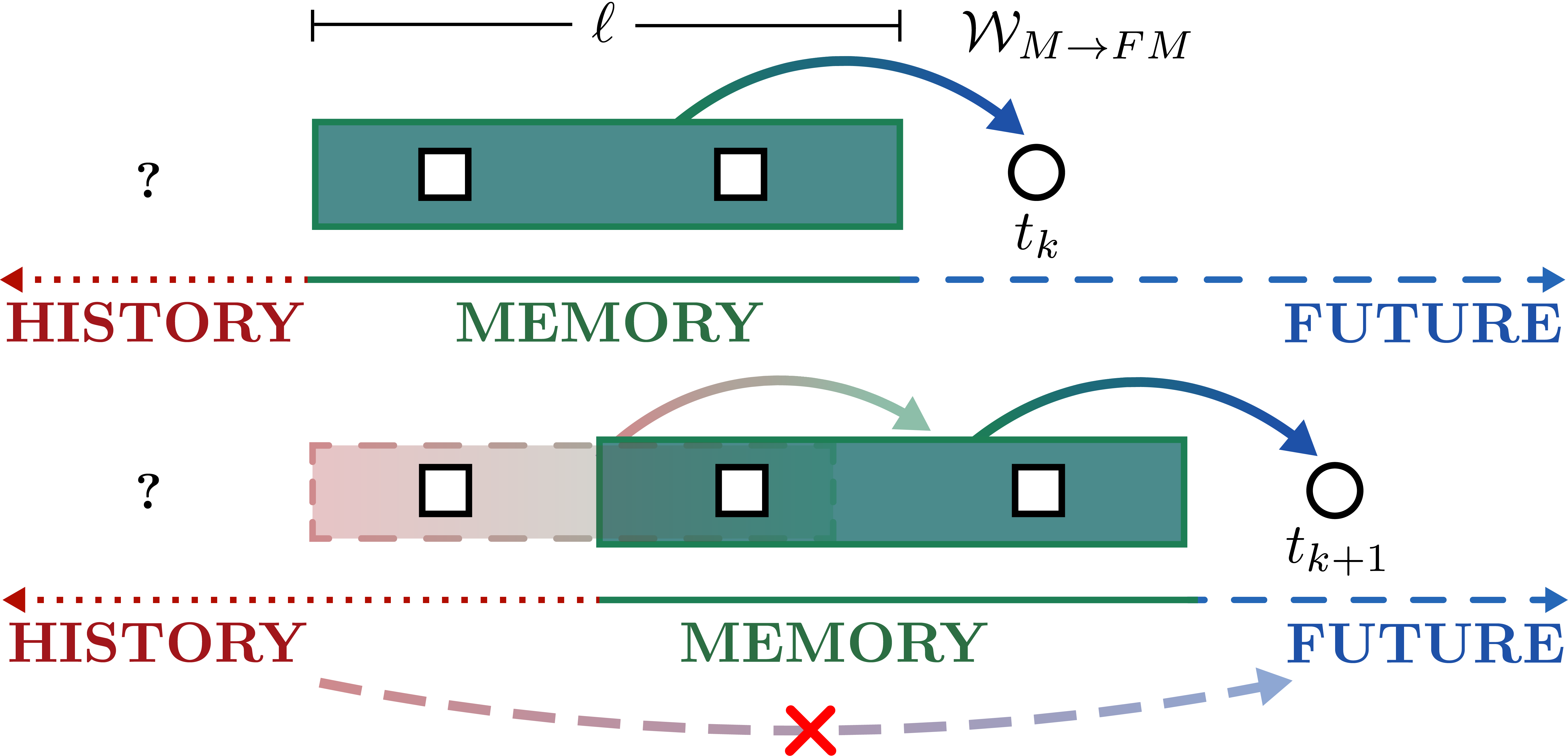}
\caption{Knowledge of the $\ell=2$ states in the memory block suffices to predict the future statistics via the recovery map $\mathcal{W}_{M \to FM}$. No information about the history is required to determine these probabilities, indicated by the question mark. This property holds true for all timesteps. Any influence the history (beyond $\ell$ timesteps ago) has on the future is mediated through memory blocks; conditional on the states in the most recent block, there are no correlations between history and future, indicated by the dashed arrow.} \label{fig:memory}
\end{figure}


Operationally, the significance of finite memory length is best encapsulated through the notion of a recovery map, $\mathcal{W}_{M \to FM}$, which acts only on $M$ to give the correct future statistics: $\mathbbm{P}_{FMH}(x_F, x_M, x_H) = \mathcal{W}_{M \to FM} [\mathbbm{P}_{MH}(x_M, x_H)]$. The complexity of any predictive model is fundamentally bounded by the length of the block $M$ on which it acts (as well as the number of possible realizations of each $X_j$). Recently, the recovery map has featured in the quantum information literature: here, quantum Markov chains are defined as \textit{states} with vanishing quantum CMI~\cite{Ruskai2002, Petz2003}, or, equivalently, those satisfying quantum generalizations of recoverability~\cite{Hayden2004, Ibinson2008, Fawzi2015, Wilde2015, Sutter2016, Sutter2017}. However, it is unclear how such characterizations relate to \emph{temporal processes}, where one has access to an evolving quantum system across multiple times. We now consider a framework that provides the most general description of quantum (and classical) stochastic processes, allowing us to extend the concept of Markov order to quantum processes.

\textit{Quantum Stochastic Processes.}---A classical stochastic process is characterized by the joint probability distribution $\mathbbm{P}(x_n,t_n; \ldots; x_2,t_2; x_1,t_1)$ over events at different times, expressing, \emph{e.g.}, the probability for a molecule to be found in region $x_1$ at time $t_1$ \emph{and} region $x_2$ at $t_2$, and so forth. Analogously, a quantum stochastic process can be considered as a set of joint probability distributions for a sequence of measurement outcomes. However, in contrast, there is a continuous family of possible measurements and the choice of measurement at one time (or even whether to measure at all) can affect future statistics~\cite{Milz2017KET,Pollock2018A,Pollock2018L}. To account for this, we separate the \emph{controllable} influence applied to the system corresponding to a set of measurement outcomes, $\{ x \}$, given a choice of experimental instrument, $\mathcal{J}$, from the \emph{uncontrollable} underlying process. 

On the controllable side, the experimenter can apply any valid transformation to the system at each timestep, formally described by an \emph{instrument}. Mathematically, an instrument is a collection of completely-positive (\textbf{CP}) maps $\mathcal{J}_j = \{ O_j^{(x_j)}\}$ that chronicle the transformation the system undergoes upon realization of each measurement outcome. Specifying an instrument at timestep $t_j$ allows an experimenter to observe outcome $x_j$ with probability $\mathbbm{P}_j(x_j|\mathcal{J}_j)$. The average transformation effected by the instrument is given by the completely-positive trace-preserving (\textbf{CPTP}) map $O_j^{\mathcal{J}_j}:=\sum_{x_j}O_j^{(x_j)}$. Without loss of generality, we use the Choi-Jamio{\l}kowski isomorphism to represent all such maps as bipartite quantum states~\cite{Milz2017,Jamiolkowski1972,Choi1975}. Moreover, instruments can be extended across multiple timesteps, describing correlated measurements and repeated interactions with an ancilla; implementing such an \emph{instrument sequence} yields the joint statistics $\mathbbm{P}_{n:0}(x_{n:0}|\mathcal{J}_{n:0})$. The corresponding correlated transformations to the system associated to observing a sequence $x_{n:0}$ can be represented as a many-body Choi state $O_{n:0}^{(x_{n:0})}$. 

Crucially, to describe quantum stochastic processes, one must distinguish between such instruments and the underlying process. The former constitutes all that is controllable by an experimenter, while the latter stems from the uncontrollable system-environment dynamics. The process itself is encapsulated in the \emph{process tensor}, $\Upsilon_{n:0}$, whose Choi state is a multipartite density operator, naturally generalizing the joint probability distributions that characterize classical stochastic processes~\cite{Milz2018}. The process tensor is a linear map taking any sequence of transformations $O_{n:0}^{(x_{n:0})}$ to the corresponding joint probability distribution of its realization; it thus contains \emph{all} multi-time probabilities deducible by \emph{all} possible instrument sequences, calculated via:
\begin{align}\label{eq:processtensor}
    \mathbbm{P}_{n:0}(x_{n:0}|\mathcal{J}_{n:0}) = \tr{ O_{n:0}^{(x_{n:0})} \Upsilon_{n:0}}.
\end{align}
This is a temporal generalization of the Born rule~\cite{ShrapnelCosta2017}, and is directly analogous to its spatial counterpart, which relates observed statistics to measurement operators (instead of instruments) and a density operator (instead of a process tensor). That such an object exists is a consequence of the linearity of quantum mechanics; like the density operator, the process tensor can be (and has been~\cite{Ringbauer2015}) tomographically constructed in a finite number of experiments. In anticipation of our main results, we emphasize that Eq.~\eqref{eq:processtensor} can be used to deduce conditional processes given specification of outcomes over a subset of timesteps by restricting the trace appropriately (see Appendix~\ref{app:processtensor} for further details on the process tensor formalism). 

The process tensor extends the CPTP map paradigm to capture multi-time effects, and can be simulated using open systems techniques capable of computing multi-time correlations~\cite{Pollock2018A}. Once known, it allows one to, \emph{e.g.}, calculate the system density operator at each timestep; but, importantly, also includes all multi-time correlations, providing the most general description of open dynamics within quantum and classical physics. Notably, this approach has been used to prove that quantum processes satisfy a generalized Kolmogorov extension theorem~\cite{Milz2017KET}, thereby allowing joint and conditional probability distributions to be calculated. Crucially, we can now meaningfully construct quantum generalizations of Def.~\ref{def:cmarkovorder}, granting a fundamental study of memory in quantum processes.


\textit{Quantum Markov Order.}---For any fixed choice of instruments used to probe a quantum process, one yields a probability distribution describing a classical stochastic process~\cite{Pollock2018L}. A natural approach to extending Markov order to quantum processes is to demand such classical processes satisfy Def.~\ref{def:cmarkovorder}; however, each choice of instruments generally leads to statistics describing different classical processes. Nonetheless, a sensible requirement of a quantum process with finite-length memory is that \emph{any} future statistics deducible (no matter which future instruments are chosen) are conditionally independent of \emph{any} historical statistics, given knowledge of a length-$\ell$ instrument sequence on the memory. We define quantum Markov order accordingly:
\begin{defn}\label{def:qmarkovorder} \textbf{(Quantum Markov Order)}
    A quantum stochastic process has Markov order-$\ell$ with respect to a family of instruments $\{\mathcal{J}_M\}$ when the statistics deducible from the process satisfy Def.~\ref{def:cmarkovorder}:
    \begin{align}\label{eq:qmarkovorder1}
        \hspace{-0.25em}\!\mathbbm{P}_{F}(x_{F}|\mathcal{J}_{F};x_{M},\mathcal{J}_{M};x_{H},\mathcal{J}_{H})\!=\!\mathbbm{P}_{F}(x_{F}| \mathcal{J}_{F}; x_{M},\mathcal{J}_{M}),\!
    \end{align}
    for each $\mathcal{J}_M$ and for all possible history and future instruments $\mathcal{J}_H$ and $\mathcal{J}_F$.
\end{defn}

Intuitively, this means that for any future instruments one might apply to the system, the statistics of different measurement outcomes are determined by the most recent $\ell$ instruments and outcomes. Equivalently, given specification of the outcomes of the past $\ell$ instruments, the process governing the future dynamics is uncorrelated with that of the history, guaranteeing that any deducible statistics on the history and future are independent. Def.~\ref{def:qmarkovorder} leads to the following product structure condition on the process tensor (see Appendix~\ref{app:definition}): 
\begin{gather}\label{eq:qmarkovorderchoi}
    \Upsilon_{FH}^{(x_M)} := \ptr{M}{ O_M^{(x_M)} \Upsilon_{FMH}}
    = \Upsilon_F^{(x_M)} \otimes \Upsilon_H^{(x_M)} 
\end{gather}
for all $O_M^{(x_M)}\in \mathcal{J}_M$. In analogy to Eq.~(\ref{eq:cmarkovcondindep}), the conditional history and future processes are independent for each realization of the instrument applied. Importantly, Def.~\ref{def:qmarkovorder} reduces to Def.~\ref{def:cmarkovorder} in the correct limit:

\begin{thm}\label{thm:classicallimit}
When restricted to classical stochastic processes, Def.~\ref{def:qmarkovorder} reduces to Def.~\ref{def:cmarkovorder} for any choice of (sharp) classical instruments.
\end{thm}

\begin{proof}
It suffices to show that when restricted to probing a classical stochastic process with sharp classical instruments, $\mathbbm{P}_{FMH}(x_F,x_M,x_H|\mathcal{J}^{\text{cl}}_F, \mathcal{J}^{\text{cl}}_M, \mathcal{J}^{\text{cl}}_{H}) = \mathbbm{P}_{FMH}(x_F,x_M,x_H)$. Sharp classical instruments correspond to a complete set of (rank-1) projections onto orthogonal states at each timestep: $\mathcal{J}^{\text{cl}} = \{ \Pi_\inp^{(x)} \otimes \Pi_\out^{(x)}\}$, where $\Pi_\inp^{(x)} = \Pi_\out^{(x)} := \ket{x} \bra{x}$ satisfy $\tr{ \Pi^{(x)} \Pi^{(y)} } = \delta_{xy} \; \forall \; x, y$ ($\inp$ and $\out$ refer, respectively, to the input and output spaces of the maps applied at each timestep). A stochastic process arising from classical physics, \emph{i.e.}, a joint probability distribution, can be encoded in the diagonal of a process tensor with respect to the local product basis that the measurements act in; thus, the process tensor has the structure: $\Upsilon_{FMH}^{\text{cl}} = \sum \mathbbm{P}^\inp_{F M H}(y_F, y_M, y_H) \; \Pi_{F^\inp}^{(y_F)} \otimes \Pi_{M^\inp}^{(y_M)} \otimes \Pi_{H^\inp}^{(y_H)} \otimes \mathbbm{1}^\out_{F M H}$. Evaluating $\mathbbm{P}_{FMH}(x_F,x_M,x_H|\mathcal{J}_F^{\text{cl}},\mathcal{J}^{\text{cl}}_M,\mathcal{J}^{\text{cl}}_H)$ according to Eq.~(\ref{eq:processtensor}) yields $\mathbbm{P}_{FMH}(x_{F},x_M,x_H)$.
\end{proof}

Thus, our definition generalizes Markov order to quantum mechanics. However, for quantum processes, a much richer arsenal of instruments can be implemented. Demanding that Eq.~(\ref{eq:qmarkovorder1}) holds for arbitrary instrument sequences trivializes the theory:

\begin{thm}\label{thm:nogo}
The only quantum processes with finite Markov order with respect to all possible instruments are Markovian.
\end{thm}

The proof is given in Appendix~\ref{app:proof} and uses the fact that the set of CP maps forms a vector space to show that the only processes satisfying Eq.~\eqref{eq:qmarkovorderchoi} for all instruments have trivial Markov order. Specifically, we show that if a process has finite Markov order for a complete basis of CP maps on $M$, it cannot have finite Markov order with respect to any linear combination of them. This implies the following property:

\begin{rem*}
Any non-Markovian quantum process has infinite Markov order with respect to a generic instrument sequence. 
\end{rem*}

In light of this finding, it is clear that the classical Markov order statement in Def.~\ref{def:cmarkovorder} is weak, as it does not consider \emph{how} one measures outcomes and assumes the ability for sharp observations. Indeed, when one allows for noisy classical measurements, the product structure of Eq.~(\ref{eq:cmarkovcondindep}) breaks down, even for Markovian processes~\cite{Kleinhans2007}. This is not due to intrinsic memory of the process; rather, due to information about the history leaking into the future, thanks to the fuzziness of the measurements. Although this issue is liftable in classical physics, in quantum mechanics it is fundamental: even sharp quantum measurements appear noisy as they do not generally reveal the full state of the system. In contrast to the classical statement, demanding Def.~\ref{def:qmarkovorder} to hold for all instruments is strong, requiring the observed statistics to satisfy the Markov order-$\ell$ property \emph{no matter} how they are measured; Theorem~\ref{thm:nogo} shows that this is too restrictive. This result motivates the following introduction of a relaxed, instrument-specific definition for quantum Markov order.

\textit{Instrument-specific Quantum Markov Order.}---We say that a stochastic process has quantum Markov order-$\ell$ \emph{with respect to the instrument sequence} $\mathcal{J}_{M}$ when Def.~\ref{def:qmarkovorder} is satisfied for each realization of the sequence in question. In terms of the process tensor, this implies that there exists an instrument such that Eq.~\eqref{eq:qmarkovorderchoi} is satisfied. Importantly, whilst the instrument on the memory block must be specified, the history and future instruments remain arbitrary: for each realization of the memory instrument, any deducible statistics on the history and future are conditionally independent. This is illustrated in Fig.~\ref{fig:erasure}, where the transformations $O_M^{(x_M)}$ that `break apart' the process are temporally correlated (as they will be generically). Interestingly, quantum processes with finite Markov order have starkly distinct properties from their classical counterparts.

\begin{prop}
In contrast to classical processes, quantum processes with finite instrument-specific quantum Markov order can have non-vanishing quantum CMI.
\end{prop}

\begin{proof}
Consider $\Upsilon_{FMH} = \sum_{x} \mathbbm{P}(x) \Upsilon_F^{(x)} \otimes \Delta_M^{(x)} \otimes \Upsilon_H^{(x)}$ such that $\Upsilon_{FMH} \geq 0$ and $\tr{\Delta_M^{(x)} O_M^{(y)}} = \delta_{xy} \, \forall \, x,y$, where each $O_M^{(y)}$ is an element of some instrument $\mathcal{J}_M = \{ O_M^{(y)} \}$. Such processes can have non-vanishing quantum CMI, $I(F:H|M) > 0$ when the Choi states of the $O_M^{(y)}$ do not all commute; indeed, $I(F:H|M)$ is not monotonic with respect to instruments in $M$, and is therefore a poor quantifier for memory strength. Nonetheless, such processes have finite Markov order with respect to the instrument $\mathcal{J}_M$ (see Appendix~\ref{app:example}).
\end{proof}

As highlighted above, quantum Markov order permits a vast landscape of memory effects: i) the decoupling instruments can vary across timesteps (or even be necessarily correlated); ii) at each timestep, instruments need not comprise only orthogonal projectors; iii) deterministic instruments can break future-history correlations; and iv) quantum CMI is not necessarily vanishing for processes with finite quantum Markov order. In an accompanying Article~\cite{Taranto2018A}, we explore the structure of processes with finite quantum Markov order with respect to natural classes of instruments, shedding light on such distinguishing features. We now discuss the broader implications of our work.



\begin{figure}[t]
\centering
\includegraphics[width=\linewidth]{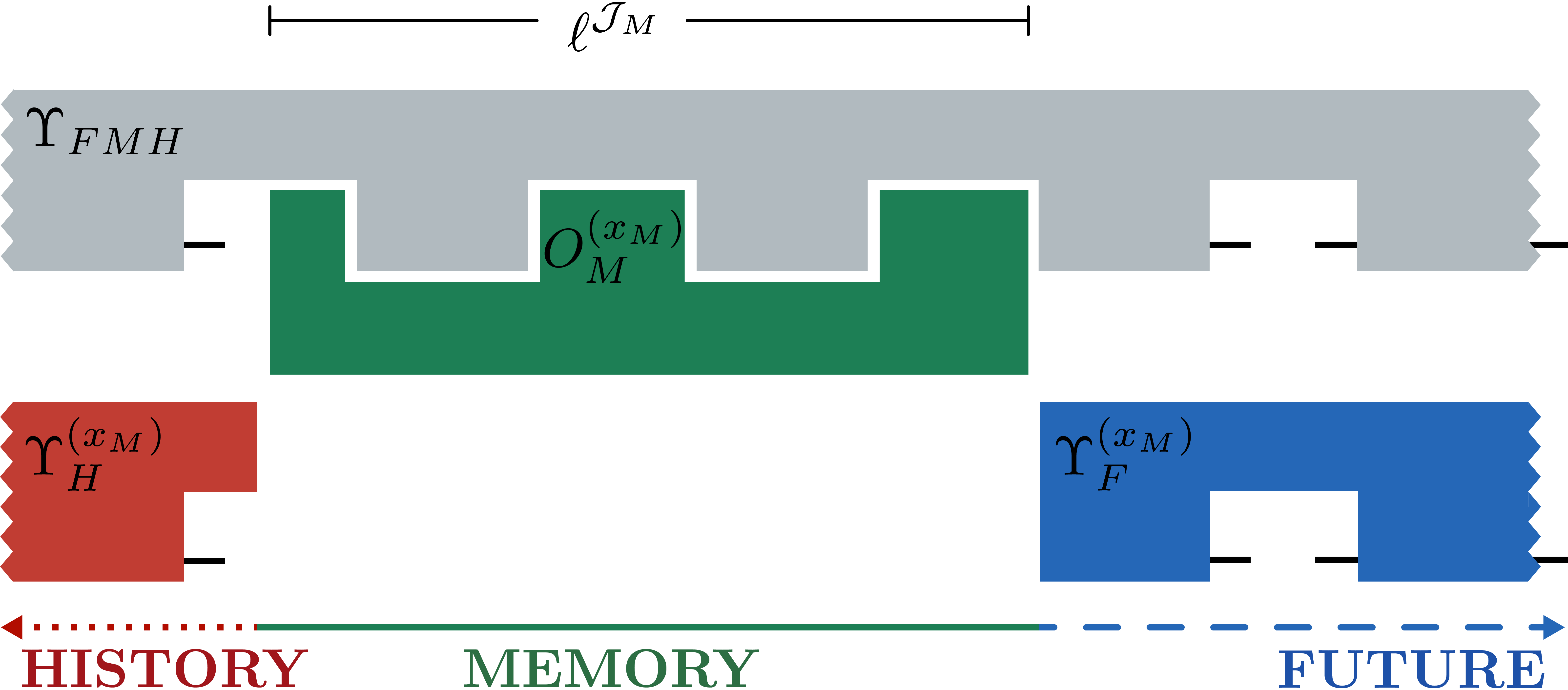}
\caption{An instrument $\mathcal{J}_M$, comprising (temporally correlated) CP maps $\{O_M^{(x_M)}\}$ (green) over $\ell$ timesteps, is applied to a process $\Upsilon_{FMH}$. The process has Markov order-$\ell$ with respect to this instrument, denoted $\ell^{\mathcal{J}_M}$ when, for each possible realization, $x_M$, the history (red) and future (blue) conditional processes are independent.} \label{fig:erasure}
\end{figure}



\textit{Conclusions}.---In this Letter we have formulated an extension of Markov order to the quantum realm, which reduces to the classical condition appropriately. Theorem~\ref{thm:nogo} shows that demanding the proposed condition to hold for all possible instruments is too strict, immediately trivializing the theory. The implication is that, generically, non-Markovian quantum processes have infinite Markov order; they exhibit distinct memory effects when probed differently. Interestingly, such instrument-specific effects have been observed~\cite{Wiesner2006}; our characterization formally explains such behavior. 

This led us to propose a relaxed definition of instrument-specific quantum Markov order. Here, the history and future processes are independent conditioned on each outcome of an instrument sequence specified on the memory. Perhaps surprisingly, when one allows for noisy measurements (or interventions more generally) in the classical setting, a similar relaxation is necessary. This has significant implications for the reconstruction of complex dynamics~\cite{Siefert2003,Bottcher2006,Kleinhans2007,Lehle2011}, highlighting that the standard formulation of Markov order is weaker than is necessary in quantum mechanics, where measurements are inherently fuzzy. 

Our framework opens the door for a comprehensive and unambiguous study of memory effects in quantum processes, which has hitherto been elusive. By capturing multi-time statistics, it goes beyond state-of-the-art descriptions which typically only consider two-point correlations and thus cannot properly describe memory effects.

Non-Markovian processes with strictly finite Markov order are unlikely to be found in nature; however, numerical techniques for open dynamics often invoke finite memory approximations, where rapidly vanishing temporal correlations are truncated~\cite{Makri1995JCP-1, Makri1995JCP-2, deVega2017}. This is tantamount to treating the process as having finite Markov order with respect to the identity instrument (\emph{i.e.}, do nothing) when the correlations considered involve operators evolved freely by the underlying Hamiltonian, although memory approximations involving other choices of instruments can also be made~\cite{Pollock2018Q}.

Understanding memory effects has immediate relevance to developing near-term quantum technologies, particularly concerning the construction of error-correcting codes to combat correlated noise~\cite{Clemens2004,Chiribella2011,Ben-Aroya2011,Lupo2012,Edmunds2017} and the design of feedback protocols for coherent control~\cite{Grimsmo2015,Whalen2017}. Our work poses the following questions for further investigation: Which instrument sequence is optimal in blocking the influence of a given process's history? What constraints are imposed on the underlying system-environment dynamics? How can we measure memory strength to characterize processes with approximately finite memory, and what are the subsequent implications for recoverability? These questions, among others, are critical for both our foundational understanding of quantum theory and the efficient simulation of quantum processes. 


\begin{acknowledgments}
We thank Top Notoh for insightful discussions. P. T. is supported by the Australian Government Research Training Program Scholarship and the J. L. William Scholarship. S. M. is supported by the Monash Graduate Scholarship, the Monash International Postgraduate Research Scholarship and the J. L. William Scholarship. K. M. is supported through the Australian Research Council Future Fellowship FT160100073.
\end{acknowledgments}


\appendix 

\section{Introduction to Process Tensor}\label{app:processtensor}

The process tensor has been derived in a variety of contexts. In this Appendix, we introduce some important concepts from the formalism that should aid the reader in following the more technical details of the Supplemental Material of this Letter. For a more thorough and pedagogical account, we refer the reader to the following developments of the framework: Refs.~\cite{Pollock2018L,Pollock2018A,Milz2017} for an operational open systems dynamics perspective; Refs.~\cite{Chiribella2008,Chiribella2008-2,Chiribella2009} for derivation by way of general quantum circuit architectures; and Refs.~\cite{Oreshkov2012,Araujo2015,Oreshkov2016,ShrapnelCosta2017} for the process matrix approach. Here, we follow the latter one for its clear connection to measurable joint probabilities.

Quantum theory is, at its core, a probabilistic theory about measurement statistics observed through experiment. It is well-known that quantum mechanics cannot be adequately described within standard probability theory due to inherently non-classical features, such as non-commutativity and contextuality. For example, we are required to introduce positive operator-valued measures \textbf{(POVMs)} that associate an operator to each outcome of a measurement, with the Born rule providing the appropriate assignment of probabilities to measurement outcomes. In the standard setting, a POVM is a set of positive semi-definite operators $\mathcal{J} = \{ \Pi^{(x)} \}$ that sum to the identity $\sum_x \Pi^{(x)} = \mathbbm{1}$, with the conditions reflecting the fact that each possible outcome can occur with some non-negative probability and some outcome occurs, respectively. Each element of the POVM corresponds to an \emph{event}, \emph{i.e.}, a realization of the associated measurement outcome. Quantum states are defined as linear functionals on the POVM elements, represented by positive semi-definite, unit-trace, Hermitian operators $\rho \in \mathcal{B}(\mathcal{H})$. The Born rule describes the probability for measuring a certain outcome $x$ given a quantum system described by some density operator $\rho$ is interrogated by the measurement instrument $\mathcal{J}$:
\begin{align}\label{eq:born}
    \mathbbm{P}(x|\mathcal{J}) = \tr{\Pi^{(x)} \rho}.
\end{align}

However, the Born rule does not assign joint probabilities to consecutive events; in order to study temporal processes, we must track the transformations of the system over time upon observation of outcomes, which cannot be accounted for by POVMs. The generalization of a POVM that captures this information is an instrument, which reduces to a POVM for a trivial output space. Formally, an instrument is a collection of CP maps that sum to a CPTP map, with the events of the theory now elevated to correspond to the individual CP elements. All such maps must act linearly to preserve the statistical nature of quantum theory, and, as such, can be represented as bipartite operators through the CJI~\cite{Jamiolkowski1972,Choi1975}. Consider a CP map taking input states to subnormalized output states $\mathcal{O} : \mathcal{B}(\mathcal{H}^\inp) \to \mathcal{B}(\mathcal{H}^\out)$. The CJI associates to any such linear map an operator $O \in \mathcal{B}(\mathcal{B}(\mathcal{H}^\out) \otimes \mathcal{B}(\mathcal{H}^\inp))$ via its action on half of an unnormalized maximally entangled state: $O := \mathcal{O} \otimes \mathcal{I} (\Phi)$, where $\Phi$ is the projector of $\sum_{i=1}^{d^\inp} \ket{ii}$ and $\mathcal{I}$ is the identity map. In this representation, complete-positivity of the map corresponds to a positive semi-definite Choi state and trace-preservation demands $\ptr{\out}{O} = \mathbbm{1}^\inp$. Thus, an instrument is a collection of CP maps $\mathcal{J} = \{ O^{(x)} \}$, each living in $\mathcal{B}(\mathcal{B}(\mathcal{H}^\out) \otimes \mathcal{B}(\mathcal{H}^\inp))$, such that $O^{\mathcal{J}} := \sum_x O^{(x)}$ satisfies $\ptr{\out}{O^{\mathcal{J}}} = \mathbbm{1}^\inp$. The Choi state of a CPTP map therefore has trace equal to $d^\inp$.

In the most general setting, instruments can be correlated in time, which means they need not be restricted to sequences of individual instruments, but can rather be collections of so-called \emph{tester elements}~\cite{Chiribella2009}, which we refer to as an \emph{instrument sequence}. Formally, a tester or instrument sequence is a collection $\{ O_{n:0}^{(x_{n:0})} \}$ of correlated transformations living on $\mathcal{B}(\bigotimes_{j=0}^n \mathcal{B}\left(\mathcal{H}^\out_j) \otimes \mathcal{B}(\mathcal{H}^\inp_j)\right)$, which, when summed over all possible outcomes, yields an overall deterministic proper process, as we soon define. 

Consider now an $(n+1)$--step process probed by a sequence of instruments with Choi states $\{ O^{(x_j)}_j \}$, where the subscript denotes the timestep. In order to yield joint probabilities in time, one can define a multi-linear functional over sets of events (\emph{i.e.}, sequences of CP maps); this is precisely the process tensor, $\Upsilon_{n:0}$. Specifying an element $O_{n:0}^{(x_{n:0})}$ of any valid instrument sequence, the following generalized spatio-temporal Born rule maps any such sequence to the joint probability for it to be realized~\footnote{In the original formulation, the process tensor also outputs the final density operator of the system~\cite{Pollock2018A}, which is not important to our analysis here. Further, note that Eq.~\eqref{eq:generalizedborn} (Eq.~\eqref{eq:processtensor} in the main text) is correct up to a transpose of $O_{n:0}^{(x_{n:0})}$; we do not explicitly indicate this transposition to ease notation, with no affect on any results.}:
\begin{align}\label{eq:generalizedborn}
    &\mathbbm{P}(x_{n:0} | \mathcal{J}_{n:0}) = \tr{ O_{n:0}^{(x_{n:0})} \Upsilon_{n:0}}.
\end{align}
In clear analogy to Eq.~\eqref{eq:born}, here, elements of the instrument sequence play the role of POVM elements, with a sequence of events in time constituting a realization or sample trajectory of the process. Of course, instrument sequences include as a special case deterministic transformations with a single `measurement outcome', such as unitary operations. Applying the identity map at all preceding timesteps yields the density operator of the system at each point in time throughout the dynamics. 

However, crucially, the process tensor description captures much more than the system density operator at each timestep. Indeed, the process tensor generalizes the notion of the density operator, inasmuch as it contains sufficient information to deduce all possible joint probabilities corresponding to the realization of any valid instrument sequence applied, thereby providing an operationally meaningful characterization of a quantum stochastic process. Due to the linearity of quantum mechanics, the process tensor can be probed and experimentally reconstructed in a finite number of experiments, and can be represented as a many-body Choi operator living in $\mathcal{B}(\bigotimes_{j=0}^n \mathcal{B}\left(\mathcal{H}^\out_j) \otimes \mathcal{B}(\mathcal{H}^\inp_j)\right)$. Like the density operator, the process tensor must satisfy certain constraints to ensure that the generalized Born rule above yields a valid probability distribution for all possible instrument sequences. In addition, a natural causality constraint must be satisfied to ensure that the choice of instruments in the future do not affect the statistics observed in the history. In summary, a proper process tensor must be a positive semi-definite, Hermitian operator such that:
\begin{align}\label{eq:causality}
    \ptr{j^\inp}{\Upsilon_{j:0}} = \mathbbm{1}^\out_j \otimes \Upsilon_{j-1:0} \; \forall \, j = 1, \hdots, n. 
\end{align}
The satisfaction of the above hierarchy of linear constraints leads to the fact that the Choi state of a process must have trace equal to the product of the dimension of the system on all of its output Hilbert spaces: $d^\out_{n:0} := d_0^\out \times \hdots \times d_n^\out$. Note the complementarity of the labeling convention regarding the input/output spaces of the process tensor and the instruments applied. These spaces are defined with respect to the perspective of an experimenter probing the process; thus, what is `output' by the process constitutes an `input' for the experimenter. Any valid instrument sequence must therefore satisfy the complementary set of trace conditions to those of Eq.~\eqref{eq:causality}.

Lastly, note that by restricting the trace of Eq.~\eqref{eq:generalizedborn} to a subset of timesteps over which an instrument sequence has been specified, one can deduce the conditional process defined on the remaining timesteps. For example, suppose that an experimenter measures the outcome corresponding to the operator $\Pi^{(x_0)}_{0^\inp}$, followed by a deterministic re-preparation of the state $\rho_{0^\out}$ at the first timestep. Then, the conditional process defined over the remaining timesteps is:
\begin{align}\label{eq:conditionalprocess}
    \Upsilon_{n:1}^{(x_0)} = \ptr{0}{\mathbbm{1}_{n:1} \otimes \rho_{0^\out} \otimes \Pi^{(x_0)}_{0^\inp} \Upsilon_{n:0}} / \, \mathbbm{P}(x_0 | \mathcal{J}_0).
\end{align}
The object $\Upsilon_{n:1}^{(x_0)}$ satisfies the conditions of Eq.~\eqref{eq:causality} and is therefore a proper process, representing the conditional process defined on the timesteps $t_1, \hdots, t_n$ with respect to knowledge that outcome $x_0$ was observed at the first timestep.

\section{Process Tensor Constraint from Def.~\ref{def:qmarkovorder}.}\label{app:definition}

In this Appendix, we show how Eq.~\eqref{eq:qmarkovorderchoi} follows from Def.~\ref{def:qmarkovorder} and Eq.~\eqref{eq:processtensor} when we allow the history and future instruments to remain arbitrary. Firstly, writing out the conditioning and marginalization in Eq.~\eqref{eq:qmarkovorder1} explicitly, we have:
\begin{widetext}
\begin{align}
    \frac{\mathbbm{P}_{FMH}(x_F,x_M,x_H|\mathcal{J}_F,\mathcal{J}_M,\mathcal{J}_H)}{\sum_{x_F}\mathbbm{P}_{FMH}(x_F,x_M,x_H|\mathcal{J}_F,\mathcal{J}_M,\mathcal{J}_H)} &=   \frac{\sum_{x_H}\mathbbm{P}_{FMH}(x_F,x_M,x_H|\mathcal{J}_F,\mathcal{J}_M,\mathcal{J}_H)}{\sum_{x_F x_H}\mathbbm{P}_{FMH}(x_F,x_M,x_H|\mathcal{J}_F,\mathcal{J}_M,\mathcal{J}_H)};
\end{align}
substituting in Eq.~\eqref{eq:processtensor} then leads to:
\begin{align} \label{eq:conditionalprobprocesstensor}
    \frac{\tr{ \left(O_{F}^{(x_F)}\otimes O_{M}^{(x_M)}\otimes O_{H}^{(x_H)}\right) \Upsilon_{FMH}}}{\tr{\left( O_{F}^{\mathcal{J}_F}\otimes O_{M}^{(x_M)}\otimes O_{H}^{(x_H)} \right) \Upsilon_{FMH}}} &= \frac{\tr{ \left( O_{F}^{(x_F)}\otimes O_{M}^{(x_M)}\otimes O_{H}^{\mathcal{J}_H}\right) \Upsilon_{FMH}}}{\tr{\left( O_{F}^{\mathcal{J}_F}\otimes O_{M}^{(x_M)}\otimes O_{H}^{\mathcal{J}_H}\right) \Upsilon_{FMH}}},
\end{align}
\end{widetext}
where we have use the notation introduced in the main text to indicate the average CPTP map $O_{X}^{\mathcal{J}_X}=\sum_{x_X}O_{X}^{(x_X)}$ applied when the instrument $\mathcal{J}_X$ is chosen. We can further simplify the expression by the causal structure of the process tensor (see Eq.~\eqref{eq:causality}) which states that the statistics of the past cannot be influenced by the choice of instrument in the future. This means that $\tr{\left( O_{F}^{\mathcal{J}_F}\otimes O_{M}^{(x_M)}\otimes O_{H}^{(x_H)}\right) \Upsilon_{FMH}}=\tr{ \left(O_{M}^{(x_M)}\otimes O_{H}^{(x_H)}\right) \Upsilon_{MH}}$ for any $\mathcal{J}_F$, where $\Upsilon_{MH}=\ptr{F}{\Upsilon_{FMH}}$. 
Defining the conditional process tensor:
\begin{align}
    \Upsilon_F^{(x_M,x_H)}:=&\frac{\ptr{MH}{\left(\mathbbm{1}_F\otimes O_{M}^{(x_M)}\otimes O_{H}^{(x_H)} \right) \Upsilon_{FMH}}}{\mathbbm{P}(x_M,x_H|\mathcal{J}_M,\mathcal{J}_H)},\notag
\end{align}
with:
\begin{align}
   \mathbbm{P}(x_M,x_H|\mathcal{J}_M,\mathcal{J}_H)\!=\! \tr{\!\left(\mathbbm{1}_F\otimes O_{M}^{(x_M)}\otimes O_{H}^{(x_H)} \right)\!\Upsilon_{FMH}\!},\!\notag
\end{align}
and requiring Eq.~\eqref{eq:conditionalprobprocesstensor} to hold for any $\mathcal{J}_F$ then implies that:
\begin{align}
    \Upsilon_F^{(x_M,x_H)} = \frac{\sum_{x'_H}\mathbbm{P}(x_M,x'_H|\mathcal{J}_M,\mathcal{J}_H)\Upsilon_F^{(x_M,x'_H)}}{\sum_{x''_H}\mathbbm{P}(x_M,x''_H|\mathcal{J}_M,\mathcal{J}_H)}.\label{eq:allJF}
\end{align}
Further requiring this to hold for all $\mathcal{J}_H$ leads to a contradiction unless $\Upsilon_F^{(x_M,x_H)}=\Upsilon_F^{(x_M,x'_H)}=:\Upsilon_F^{(x_M)}\; \forall \, x_H,x'_H$. This can be seen by leaving the instrument fixed but varying the outcome, changing the left hand side of Eq.~\eqref{eq:allJF} but not the right; conversely, fixing the CP map corresponding to the outcome $x_M$ and varying the overall instrument, leading the right hand side changing but not the left.

Rearranging Eq.~\eqref{eq:conditionalprobprocesstensor} and taking only a partial trace over $M$ under these conditions, thereby leaving $\mathcal{J}_F$ and $\mathcal{J}_H$ unspecified, leads to Eq.~\eqref{eq:qmarkovorderchoi} in the main text, with:
\begin{align}\label{eq:condfuture}
    \Upsilon_{F}^{(x_M)} := \frac{d^\out_F}{\alpha^{(x_M)}} \ptr{MH}{O_M^{(x_M)} \Upsilon_{FMH}}
\end{align}
and:
\begin{align}\label{eq:condhistory}
    \Upsilon_{H}^{(x_M)} := \frac{1}{d^\out_F} \ptr{FM}{O_M^{(x_M)} \Upsilon_{FMH}}.
\end{align}
Here, in defining the conditional future process in Eq.~\eqref{eq:condfuture} we have introduced a proportionality constant that depends on $x_M$, defined as: $\alpha^{(x_M)} := \tr{O_M^{(x_M)} \Upsilon_{FMH}}$. By construction, $\Upsilon_{F}^{(x_M)}$ is a proper process tensor for each $x_M$, \emph{i.e.}, a positive semi-definite operator satisfying the causality constraints of Eq.~\eqref{eq:causality}. As such, we know what its trace must be and can normalize accordingly with the appropriate constant. The expression for the conditional history process in Eq.~\eqref{eq:condhistory} is, in general, not a proper process tensor, since it need not satisfy the causality constraint as realizing a sequence of outcomes on $M$ amounts to a post-selection of the history~\cite{Chiribella2008,Milz2018}. Nonetheless, by the normalization of the total probability, we know that $1 = \sum_{x_M} \mathbbm{P}_M(x_M | \mathcal{J}_M, \mathcal{J_H}) = \sum_{x_M} \tr{O^{\mathcal{J}_H}_H \Upsilon_H^{(x_M)}} \, \forall \, O^{\mathcal{J}_H}_H$. This implies, unsurprisingly, that when summed over the possible outcomes $x_M$, the conditional historic process must yield a proper process tensor satisfying the causality constraints of Eq.~\eqref{eq:causality}, \emph{i.e.,} the conditional history process is a tester. These technicalities aside, it is the tensor product form between the conditional history and future processes that provide the important mathematical structure which coincides nicely with our physical intuition, since, upon conditioning on any outcome of the instrument sequence in question, any history and future statistics must be rendered independent.

\section{Proof of Theorem~\ref{thm:nogo}.}\label{app:proof}

We begin with the following Lemma:

\begin{lem}\label{lem:nogo}
The only operators $\Upsilon_{FMH}$ which satisfy Eq.~(\ref{eq:qmarkovorderchoi}) for all possible instruments $\mathcal{J}_M$ are those where the $M$ subsystem is in tensor product with $F$ or $H$ or both.
\end{lem}

Choose a linearly independent, informationally-complete set of projectors $\mathcal{J}=\{ \Pi_M^{(x)} \}$ as the instrument on $M$. Any linearly independent set has an associated dual set of operators $\{ \Delta_M^{(y)} \}$ such that $\tr{\Pi_M^{(x)} \Delta_M^{(y)}} = \delta_{xy} \; \forall \, x, y$~\cite{Modi2012A}. Thus, we can write any tripartite state satisfying Eq.~(\ref{eq:qmarkovorderchoi}) for each measurement outcome as follows:
\begin{align}
    \Upsilon_{FMH} = \sum_{x} \Upsilon_F^{(x)} \otimes \Delta_M^{(x)} \otimes \Upsilon_H^{(x)}.
\end{align}
Now, consider a different instrument comprising a set of projectors defined via a linear expansion of the original set $\mathcal{J}'=\{\Gamma_M^{(z)} := \sum_{x} c_{xz} \Pi_M^{(x)}\}$, with $\{c_{xz}\}$ some non-trivial coefficients. The conditional process tensor upon application of this instrument is: 
\begin{align}
    \Upsilon_{FH}^{(z)} =& \ptr{M}{\Gamma_M^{(z)} \Upsilon_{FMH}} \\
    =& \sum_{x} \Upsilon_F^{(x)} \otimes \Upsilon_H^{(x)} \, \tr{\Gamma_M^{(z)} \Delta_M^{(x)}} \notag\\
    =&\sum_{x} c_{xz} \Upsilon_F^{(x)} \otimes \Upsilon_H^{(x)}. \notag
\end{align}
This gives a conditional product state iff either $c_{xz} = \delta_{xz} \; \forall \, x, z$, which is false by construction; or, either $\Upsilon_F^{(x)}$ or $\Upsilon_H^{(x)}$ (or both) are independent of $x$. Since the original choice of linearly independent projectors was arbitrary, and we can construct any CP map from a linear combination of such elements, the proof holds for arbitrary instruments on $M$. The only remaining way to satisfy Eq.~(\ref{eq:qmarkovorderchoi}) is if either the $F$ or $H$ (or both) parts of the process tensor are in tensor product with the part on $M$.

The proof of Theorem~\ref{thm:nogo} is immediate from Lemma~\ref{lem:nogo}, once we consider the fact that the Markov order condition must hold for any block $M$ of length $\ell$. Explicitly, consider first the the block $M$ to begin at timestep $t_{k-\ell}$ and end at timestep $t_{k-1}$. Without loss of generality, suppose that, by Lemma~\ref{lem:nogo}, the process tensor factorizes into the product $\Upsilon_{n:0} = \Upsilon_{n:k-\ell} \otimes \Upsilon_{k-\ell:0}$. Had we chosen the block $M$ to begin one timestep later, the same condition leads to the product $\Upsilon_{n:0} = \Upsilon_{n:k-\ell+1} \otimes \Upsilon_{k-\ell+1:0}$. The only way for a single process to satisfy both of these conditions is if there is a CPTP channel $\Lambda_{k-\ell+1:k-\ell}$ taking the input to the process at timestep $t_{k-\ell}$ to the output at the next timestep $t_{k-\ell+1}$: $\Upsilon_{n:0} = \Upsilon_{n:k-\ell+1} \otimes \Lambda_{k-\ell+1:k-\ell} \otimes \Upsilon_{k-\ell:0}$. Repeating this argument for all timesteps of the process immediately leads to the following Markovian (product) process tensor structure: 
\begin{align}\label{eq:markovianstructure}
    \Upsilon_{n:0}^{\textup{Markov}} &= \bigotimes_{k=1}^{n} \Lambda_{k:k-1} \otimes \rho_{0}.
\end{align} 
where $\rho_{0}$ is the average initial state of the system~\cite{Pollock2018L}. \\

\section{Process with Non-Vanishing Quantum CMI.}\label{app:example}

Consider the following three-step process on a qubit system, where Alice and Bob have access to the first and second steps respectively, and the final output state is accessible to Charlie (depicted in Fig.~\ref{fig:sicpovm}). Initially, the following tripartite state is constructed: 
\begin{align}
    \rho_{ABC} = \sum_b \frac{1}{4} \rho_{A}^{(b)} \otimes \Delta_B^{(b)} \otimes \rho_C^{(b)},
\end{align}
where, for each value of $b = \{ 1, 2, 3, 4 \}$, $\Delta_B^{(b)} := \tfrac{1}{2} (\mathbbm{1} + \sqrt{3} \sum_i c_i^{(b)} \sigma_i) $ is defined in terms of Pauli matrices $\{\sigma_i\}$ with tetrahedral coefficient vectors $\{c^{(b)}\} = \{ (1,1,1), (1,-1,-1), (-1,1,-1), (-1,-1,1) \}$. These objects forms the dual set to the POVM $\mathcal{J}_B$, comprising elements $\Pi^{(b)}_B := \tfrac{1}{4} (\mathbbm{1} + \tfrac{1}{\sqrt{3}} \sum_i c_i^{(b)} \sigma_i)$. We then define the states $\rho_{X}^{(b)} = \tfrac{3}{8}  \mathbbm{1} + \tfrac{1}{2} \Pi^{(b)}$, with $X = \{ A, C\}$, in terms of these POVM elements, before finally normalizing the overall state. The process is such that the $A$ part of this state is first given to Alice, who can make any operation that she likes. After this, Alice's output is discarded and the $B$ part of the overall state is given to Bob, who can make any operation that he likes. Finally, Bob's output is discarded, and the $C$ part of the state described above is given to Charlie. The process tensor is thus: $\Upsilon_{ABC} = \rho_{ABC}^\inp \otimes \mathbbm{1}_{AB}^\out$.

Suppose Bob chooses to measure the POVM $ \Pi_B $ as his instrument. Then, Eq.~(\ref{eq:qmarkovorderchoi}) holds for each outcome and Alice's conditional state is independent of Charlie's. However, if he chooses any other instrument, Alice and Charlie's states remain correlated (at least for some outcomes). Thus, with respect to the measurement instrument $\mathcal{J}_B$, the process has Markov order~$1$, whereas it has larger Markov order for a generic instrument. Importantly, the POVM elements of Bob's measurement are non-orthogonal, so this instrument has no classical counterpart. Lastly, the quantum CMI of the process tensor does not vanish: $I(A:C|B) \approx 0.059$. Nonetheless, knowing Bob's measurement outcome with respect to the POVM $\mathcal{J}_B$ allows us to reconstruct the entire $ABC$ state and therefore unambiguously describe the process.


\begin{figure}[t]
\centering
\includegraphics[width=\linewidth]{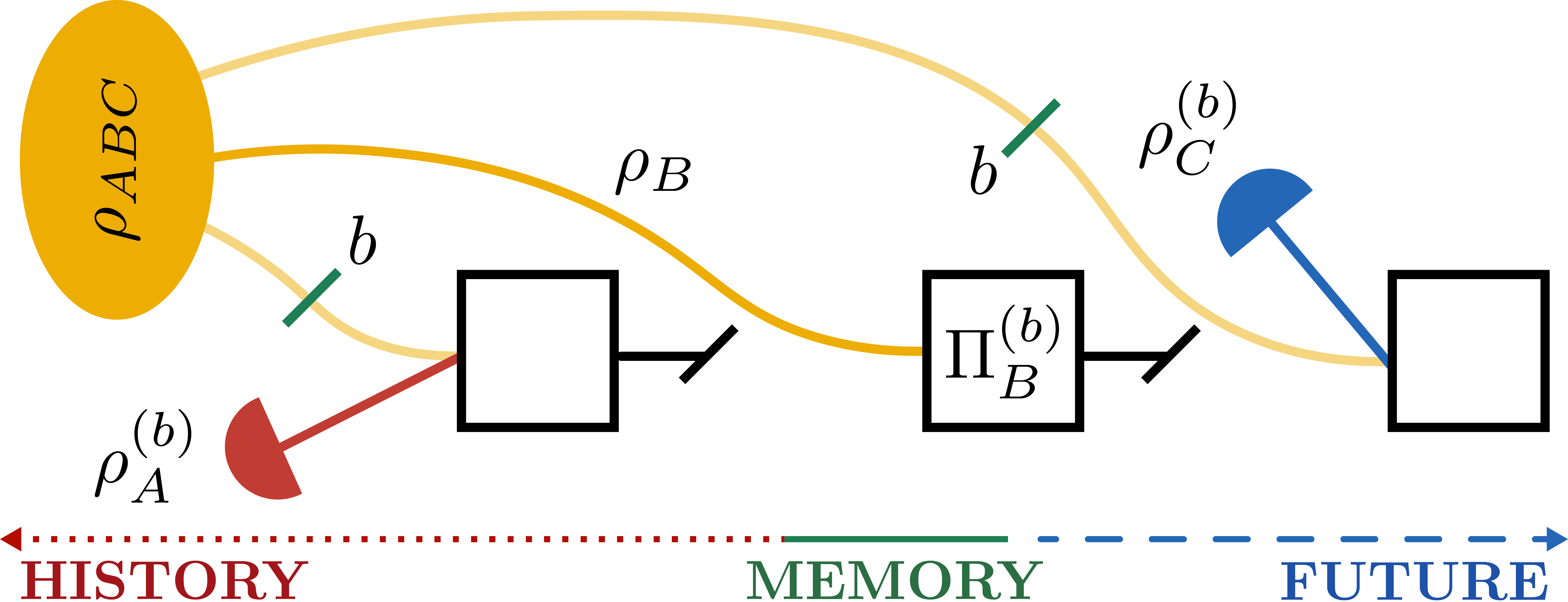}
\caption{\textit{Process with finite instrument-specific quantum Markov order but non-vanishing quantum CMI.} The process is as described in the text. Temporally, we trivialize the output spaces, so what Alice receives denotes the history (red, dotted); what Bob receives denotes the memory (green, solid); and what Charlie receives denotes the future (blue, dashed). For arbitrary instruments of Bob's choosing, Alice and Charlie's states are, in general, correlated; except for when Bob measures with the specific POVM $\mathcal{J}_B = \{ \Pi_B^{(b)}\}$. In this case, measurement of any outcome $b$ has the effect of breaking the correlations between Alice and Charlie's subsystems, rendering them in the conditional product state $\rho_{AC}^{(b)} = \rho_A^{(b)} \otimes \rho_C^{(b)}$. Importantly, Bob's instrument is distinct from anything he could implement classically and $I(A:C|B) \neq 0$. \label{fig:sicpovm}}
\end{figure}


\begin{thebibliography}{67}%
\makeatletter
\providecommand \@ifxundefined [1]{%
 \@ifx{#1\undefined}
}%
\providecommand \@ifnum [1]{%
 \ifnum #1\expandafter \@firstoftwo
 \else \expandafter \@secondoftwo
 \fi
}%
\providecommand \@ifx [1]{%
 \ifx #1\expandafter \@firstoftwo
 \else \expandafter \@secondoftwo
 \fi
}%
\providecommand \natexlab [1]{#1}%
\providecommand \enquote  [1]{``#1''}%
\providecommand \bibnamefont  [1]{#1}%
\providecommand \bibfnamefont [1]{#1}%
\providecommand \citenamefont [1]{#1}%
\providecommand \href@noop [0]{\@secondoftwo}%
\providecommand \href [0]{\begingroup \@sanitize@url \@href}%
\providecommand \@href[1]{\@@startlink{#1}\@@href}%
\providecommand \@@href[1]{\endgroup#1\@@endlink}%
\providecommand \@sanitize@url [0]{\catcode `\\12\catcode `\$12\catcode
  `\&12\catcode `\#12\catcode `\^12\catcode `\_12\catcode `\%12\relax}%
\providecommand \@@startlink[1]{}%
\providecommand \@@endlink[0]{}%
\providecommand \url  [0]{\begingroup\@sanitize@url \@url }%
\providecommand \@url [1]{\endgroup\@href {#1}{\urlprefix }}%
\providecommand \urlprefix  [0]{URL }%
\providecommand \Eprint [0]{\href }%
\providecommand \doibase [0]{http://dx.doi.org/}%
\providecommand \selectlanguage [0]{\@gobble}%
\providecommand \bibinfo  [0]{\@secondoftwo}%
\providecommand \bibfield  [0]{\@secondoftwo}%
\providecommand \translation [1]{[#1]}%
\providecommand \BibitemOpen [0]{}%
\providecommand \bibitemStop [0]{}%
\providecommand \bibitemNoStop [0]{.\EOS\space}%
\providecommand \EOS [0]{\spacefactor3000\relax}%
\providecommand \BibitemShut  [1]{\csname bibitem#1\endcsname}%
\let\auto@bib@innerbib\@empty
\bibitem [{\citenamefont {Van~Kampen}(2011)}]{StochProc}%
  \BibitemOpen
  \bibfield  {author} {\bibinfo {author} {\bibfnamefont {N.}~\bibnamefont
  {Van~Kampen}},\ }\href@noop {} {\emph {\bibinfo {title} {{Stochastic
  Processes in Physics and Chemistry}}}}\ (\bibinfo  {publisher} {Elsevier, New
  York},\ \bibinfo {year} {2011})\BibitemShut {NoStop}%
\bibitem [{\citenamefont {Breuer}\ and\ \citenamefont
  {Petruccione}(2002)}]{BreuerPetruccione}%
  \BibitemOpen
  \bibfield  {author} {\bibinfo {author} {\bibfnamefont {H.-P.}\ \bibnamefont
  {Breuer}}\ and\ \bibinfo {author} {\bibfnamefont {F.}~\bibnamefont
  {Petruccione}},\ }\href@noop {} {\emph {\bibinfo {title} {{The Theory of Open
  Quantum Systems}}}}\ (\bibinfo  {publisher} {Oxford University Press, New
  York},\ \bibinfo {year} {2002})\BibitemShut {NoStop}%
\bibitem [{\citenamefont {Salzberg}\ \emph {et~al.}(1998)\citenamefont
  {Salzberg}, \citenamefont {Delcher}, \citenamefont {Kasif},\ and\
  \citenamefont {White}}]{Salzberg1998}%
  \BibitemOpen
  \bibfield  {author} {\bibinfo {author} {\bibfnamefont {S.~L.}\ \bibnamefont
  {Salzberg}}, \bibinfo {author} {\bibfnamefont {A.~L.}\ \bibnamefont
  {Delcher}}, \bibinfo {author} {\bibfnamefont {S.}~\bibnamefont {Kasif}}, \
  and\ \bibinfo {author} {\bibfnamefont {O.}~\bibnamefont {White}},\ }\href
  {\doibase 10.1093/nar/26.2.544} {\bibfield  {journal} {\bibinfo  {journal}
  {Nucleic Acids Res.}\ }\textbf {\bibinfo {volume} {26}},\ \bibinfo {pages}
  {544} (\bibinfo {year} {1998})}\BibitemShut {NoStop}%
\bibitem [{\citenamefont {Thijs}\ \emph {et~al.}(2001)\citenamefont {Thijs},
  \citenamefont {Lescot}, \citenamefont {Marchal}, \citenamefont {Rombauts},
  \citenamefont {De~Moor}, \citenamefont {Rouze},\ and\ \citenamefont
  {Moreau}}]{Thijs2001}%
  \BibitemOpen
  \bibfield  {author} {\bibinfo {author} {\bibfnamefont {G.}~\bibnamefont
  {Thijs}}, \bibinfo {author} {\bibfnamefont {M.}~\bibnamefont {Lescot}},
  \bibinfo {author} {\bibfnamefont {K.}~\bibnamefont {Marchal}}, \bibinfo
  {author} {\bibfnamefont {S.}~\bibnamefont {Rombauts}}, \bibinfo {author}
  {\bibfnamefont {B.}~\bibnamefont {De~Moor}}, \bibinfo {author} {\bibfnamefont
  {P.}~\bibnamefont {Rouze}}, \ and\ \bibinfo {author} {\bibfnamefont
  {Y.}~\bibnamefont {Moreau}},\ }\href {\doibase
  10.1093/bioinformatics/17.12.1113} {\bibfield  {journal} {\bibinfo  {journal}
  {Bioinformatics}\ }\textbf {\bibinfo {volume} {17}},\ \bibinfo {pages} {1113}
  (\bibinfo {year} {2001})}\BibitemShut {NoStop}%
\bibitem [{\citenamefont {Rosvall}\ \emph {et~al.}(2014)\citenamefont
  {Rosvall}, \citenamefont {Esquivel}, \citenamefont {Lancichinetti},
  \citenamefont {West},\ and\ \citenamefont {Lambiotte}}]{Rosvall2014}%
  \BibitemOpen
  \bibfield  {author} {\bibinfo {author} {\bibfnamefont {M.}~\bibnamefont
  {Rosvall}}, \bibinfo {author} {\bibfnamefont {A.~V.}\ \bibnamefont
  {Esquivel}}, \bibinfo {author} {\bibfnamefont {A.}~\bibnamefont
  {Lancichinetti}}, \bibinfo {author} {\bibfnamefont {J.~D.}\ \bibnamefont
  {West}}, \ and\ \bibinfo {author} {\bibfnamefont {R.}~\bibnamefont
  {Lambiotte}},\ }\href {\doibase 10.1038/ncomms5630} {\bibfield  {journal}
  {\bibinfo  {journal} {Nat. Commun.}\ }\textbf {\bibinfo {volume} {5}},\
  \bibinfo {pages} {4630} (\bibinfo {year} {2014})}\BibitemShut {NoStop}%
\bibitem [{\citenamefont {Pollock}\ and\ \citenamefont
  {Modi}(2018{\natexlab{a}})}]{Pollock2018T}%
  \BibitemOpen
  \bibfield  {author} {\bibinfo {author} {\bibfnamefont {F.~A.}\ \bibnamefont
  {Pollock}}\ and\ \bibinfo {author} {\bibfnamefont {K.}~\bibnamefont {Modi}},\
  }\href {\doibase 10.22331/q-2018-07-11-76} {\bibfield  {journal} {\bibinfo
  {journal} {Quantum}\ }\textbf {\bibinfo {volume} {2}},\ \bibinfo {pages} {76}
  (\bibinfo {year} {2018}{\natexlab{a}})}\BibitemShut {NoStop}%
\bibitem [{\citenamefont {Nielsen}\ and\ \citenamefont
  {Chuang}(2000)}]{NielsenChuang}%
  \BibitemOpen
  \bibfield  {author} {\bibinfo {author} {\bibfnamefont {M.}~\bibnamefont
  {Nielsen}}\ and\ \bibinfo {author} {\bibfnamefont {I.}~\bibnamefont
  {Chuang}},\ }\href@noop {} {\emph {\bibinfo {title} {{Quantum Computation and
  Quantum Information}}}}\ (\bibinfo  {publisher} {Cambridge University Press,
  Cambridge, England},\ \bibinfo {year} {2000})\BibitemShut {NoStop}%
\bibitem [{\citenamefont {Viola}\ \emph {et~al.}(1999)\citenamefont {Viola},
  \citenamefont {Knill},\ and\ \citenamefont {Lloyd}}]{Viola1999}%
  \BibitemOpen
  \bibfield  {author} {\bibinfo {author} {\bibfnamefont {L.}~\bibnamefont
  {Viola}}, \bibinfo {author} {\bibfnamefont {E.}~\bibnamefont {Knill}}, \ and\
  \bibinfo {author} {\bibfnamefont {S.}~\bibnamefont {Lloyd}},\ }\href
  {\doibase 10.1103/PhysRevLett.82.2417} {\bibfield  {journal} {\bibinfo
  {journal} {Phys. Rev. Lett.}\ }\textbf {\bibinfo {volume} {82}},\ \bibinfo
  {pages} {2417} (\bibinfo {year} {1999})}\BibitemShut {NoStop}%
\bibitem [{\citenamefont {Banaszek}\ \emph {et~al.}(2004)\citenamefont
  {Banaszek}, \citenamefont {Dragan}, \citenamefont {Wasilewski},\ and\
  \citenamefont {Radzewicz}}]{Banaszek2004}%
  \BibitemOpen
  \bibfield  {author} {\bibinfo {author} {\bibfnamefont {K.}~\bibnamefont
  {Banaszek}}, \bibinfo {author} {\bibfnamefont {A.}~\bibnamefont {Dragan}},
  \bibinfo {author} {\bibfnamefont {W.}~\bibnamefont {Wasilewski}}, \ and\
  \bibinfo {author} {\bibfnamefont {C.}~\bibnamefont {Radzewicz}},\ }\href
  {\doibase 10.1103/PhysRevLett.92.257901} {\bibfield  {journal} {\bibinfo
  {journal} {Phys. Rev. Lett.}\ }\textbf {\bibinfo {volume} {92}},\ \bibinfo
  {pages} {257901} (\bibinfo {year} {2004})}\BibitemShut {NoStop}%
\bibitem [{\citenamefont {Erez}\ \emph {et~al.}(2008)\citenamefont {Erez},
  \citenamefont {Gordon}, \citenamefont {Nest},\ and\ \citenamefont
  {Kurizki}}]{Erez2008}%
  \BibitemOpen
  \bibfield  {author} {\bibinfo {author} {\bibfnamefont {N.}~\bibnamefont
  {Erez}}, \bibinfo {author} {\bibfnamefont {G.}~\bibnamefont {Gordon}},
  \bibinfo {author} {\bibfnamefont {M.}~\bibnamefont {Nest}}, \ and\ \bibinfo
  {author} {\bibfnamefont {G.}~\bibnamefont {Kurizki}},\ }\href {\doibase
  10.1038/nature06873} {\bibfield  {journal} {\bibinfo  {journal} {Nature
  (London)}\ }\textbf {\bibinfo {volume} {452}},\ \bibinfo {pages} {724}
  (\bibinfo {year} {2008})}\BibitemShut {NoStop}%
\bibitem [{\citenamefont {Reich}\ \emph {et~al.}(2015)\citenamefont {Reich},
  \citenamefont {Katz},\ and\ \citenamefont {Koch}}]{Reich2015}%
  \BibitemOpen
  \bibfield  {author} {\bibinfo {author} {\bibfnamefont {D.~M.}\ \bibnamefont
  {Reich}}, \bibinfo {author} {\bibfnamefont {N.}~\bibnamefont {Katz}}, \ and\
  \bibinfo {author} {\bibfnamefont {C.~P.}\ \bibnamefont {Koch}},\ }\href
  {\doibase 10.1038/srep12430} {\bibfield  {journal} {\bibinfo  {journal} {Sci.
  Rep.}\ }\textbf {\bibinfo {volume} {5}},\ \bibinfo {pages} {12430} (\bibinfo
  {year} {2015})}\BibitemShut {NoStop}%
\bibitem [{\citenamefont {Modi}(2011)}]{Modi2011}%
  \BibitemOpen
  \bibfield  {author} {\bibinfo {author} {\bibfnamefont {K.}~\bibnamefont
  {Modi}},\ }\href {\doibase 10.1142/S1230161211000170} {\bibfield  {journal}
  {\bibinfo  {journal} {Open Sys. Info. Dyn.}\ }\textbf {\bibinfo {volume}
  {18}},\ \bibinfo {pages} {253} (\bibinfo {year} {2011})}\BibitemShut
  {NoStop}%
\bibitem [{\citenamefont {Modi}(2012)}]{Modi2012}%
  \BibitemOpen
  \bibfield  {author} {\bibinfo {author} {\bibfnamefont {K.}~\bibnamefont
  {Modi}},\ }\href {\doibase 10.1038/srep00581} {\bibfield  {journal} {\bibinfo
   {journal} {Sci. Rep.}\ }\textbf {\bibinfo {volume} {2}},\ \bibinfo {pages}
  {581} (\bibinfo {year} {2012})}\BibitemShut {NoStop}%
\bibitem [{\citenamefont {Modi}\ \emph {et~al.}(2012)\citenamefont {Modi},
  \citenamefont {Rodr\'{\i}guez-Rosario},\ and\ \citenamefont
  {Aspuru-Guzik}}]{Modi2012A}%
  \BibitemOpen
  \bibfield  {author} {\bibinfo {author} {\bibfnamefont {K.}~\bibnamefont
  {Modi}}, \bibinfo {author} {\bibfnamefont {C.~A.}\ \bibnamefont
  {Rodr\'{\i}guez-Rosario}}, \ and\ \bibinfo {author} {\bibfnamefont
  {A.}~\bibnamefont {Aspuru-Guzik}},\ }\href {\doibase
  10.1103/PhysRevA.86.064102} {\bibfield  {journal} {\bibinfo  {journal} {Phys.
  Rev. A}\ }\textbf {\bibinfo {volume} {86}},\ \bibinfo {pages} {064102}
  (\bibinfo {year} {2012})}\BibitemShut {NoStop}%
\bibitem [{\citenamefont {Milz}\ \emph
  {et~al.}(2017{\natexlab{a}})\citenamefont {Milz}, \citenamefont {Pollock},\
  and\ \citenamefont {Modi}}]{Milz2017}%
  \BibitemOpen
  \bibfield  {author} {\bibinfo {author} {\bibfnamefont {S.}~\bibnamefont
  {Milz}}, \bibinfo {author} {\bibfnamefont {F.~A.}\ \bibnamefont {Pollock}}, \
  and\ \bibinfo {author} {\bibfnamefont {K.}~\bibnamefont {Modi}},\ }\href
  {\doibase 10.1142/S1230161217400169} {\bibfield  {journal} {\bibinfo
  {journal} {Open Syst. Inf. Dyn.}\ }\textbf {\bibinfo {volume} {24}},\
  \bibinfo {pages} {1740016} (\bibinfo {year}
  {2017}{\natexlab{a}})}\BibitemShut {NoStop}%
\bibitem [{\citenamefont {Kolmogorov}(1956)}]{Kolmogorov}%
  \BibitemOpen
  \bibfield  {author} {\bibinfo {author} {\bibfnamefont {A.~N.}\ \bibnamefont
  {Kolmogorov}},\ }\href@noop {} {\emph {\bibinfo {title} {Foundations of the
  Theory of Probability \emph{[Grundbegriffe der
  Wahrscheinlichkeitsrechnung]}}}}\ (\bibinfo  {publisher} {Chelsea, New
  York},\ \bibinfo {year} {1956})\BibitemShut {NoStop}%
\bibitem [{\citenamefont {Feller}(1971)}]{Feller}%
  \BibitemOpen
  \bibfield  {author} {\bibinfo {author} {\bibfnamefont {W.}~\bibnamefont
  {Feller}},\ }\href@noop {} {\emph {\bibinfo {title} {{An Introduction to
  Probability Theory and Its Applications}}}}\ (\bibinfo  {publisher} {Wiley,
  New York},\ \bibinfo {year} {1971})\BibitemShut {NoStop}%
\bibitem [{\citenamefont {Milz}\ \emph
  {et~al.}(2017{\natexlab{b}})\citenamefont {Milz}, \citenamefont {Sakuldee},
  \citenamefont {Pollock},\ and\ \citenamefont {Modi}}]{Milz2017KET}%
  \BibitemOpen
  \bibfield  {author} {\bibinfo {author} {\bibfnamefont {S.}~\bibnamefont
  {Milz}}, \bibinfo {author} {\bibfnamefont {F.}~\bibnamefont {Sakuldee}},
  \bibinfo {author} {\bibfnamefont {F.~A.}\ \bibnamefont {Pollock}}, \ and\
  \bibinfo {author} {\bibfnamefont {K.}~\bibnamefont {Modi}},\ }\href
  {https://arxiv.org/abs/1712.02589} {\bibfield  {journal} {\bibinfo  {journal}
  {arXiv:1712.02589}\ } (\bibinfo {year} {2017}{\natexlab{b}})}\BibitemShut
  {NoStop}%
\bibitem [{\citenamefont {Breuer}\ \emph {et~al.}(2016)\citenamefont {Breuer},
  \citenamefont {Laine}, \citenamefont {Piilo},\ and\ \citenamefont
  {Vacchini}}]{Breuer2016}%
  \BibitemOpen
  \bibfield  {author} {\bibinfo {author} {\bibfnamefont {H.-P.}\ \bibnamefont
  {Breuer}}, \bibinfo {author} {\bibfnamefont {E.-M.}\ \bibnamefont {Laine}},
  \bibinfo {author} {\bibfnamefont {J.}~\bibnamefont {Piilo}}, \ and\ \bibinfo
  {author} {\bibfnamefont {B.}~\bibnamefont {Vacchini}},\ }\href {\doibase
  10.1103/RevModPhys.88.021002} {\bibfield  {journal} {\bibinfo  {journal}
  {Rev. Mod. Phys.}\ }\textbf {\bibinfo {volume} {88}},\ \bibinfo {pages}
  {021002} (\bibinfo {year} {2016})}\BibitemShut {NoStop}%
\bibitem [{\citenamefont {Giovannetti}\ and\ \citenamefont
  {Palma}(2012)}]{Giovannetti2012}%
  \BibitemOpen
  \bibfield  {author} {\bibinfo {author} {\bibfnamefont {V.}~\bibnamefont
  {Giovannetti}}\ and\ \bibinfo {author} {\bibfnamefont {G.~M.}\ \bibnamefont
  {Palma}},\ }\href {\doibase 10.1103/PhysRevLett.108.040401} {\bibfield
  {journal} {\bibinfo  {journal} {Phys. Rev. Lett.}\ }\textbf {\bibinfo
  {volume} {108}},\ \bibinfo {pages} {040401} (\bibinfo {year}
  {2012})}\BibitemShut {NoStop}%
\bibitem [{\citenamefont {Lorenzo}\ \emph
  {et~al.}(2017{\natexlab{a}})\citenamefont {Lorenzo}, \citenamefont
  {Ciccarello},\ and\ \citenamefont {Palma}}]{Lorenzo2017A}%
  \BibitemOpen
  \bibfield  {author} {\bibinfo {author} {\bibfnamefont {S.}~\bibnamefont
  {Lorenzo}}, \bibinfo {author} {\bibfnamefont {F.}~\bibnamefont {Ciccarello}},
  \ and\ \bibinfo {author} {\bibfnamefont {G.~M.}\ \bibnamefont {Palma}},\
  }\href {\doibase 10.1103/PhysRevA.96.032107} {\bibfield  {journal} {\bibinfo
  {journal} {Phys. Rev. A}\ }\textbf {\bibinfo {volume} {96}},\ \bibinfo
  {pages} {032107} (\bibinfo {year} {2017}{\natexlab{a}})}\BibitemShut
  {NoStop}%
\bibitem [{\citenamefont {Lorenzo}\ \emph
  {et~al.}(2017{\natexlab{b}})\citenamefont {Lorenzo}, \citenamefont
  {Ciccarello}, \citenamefont {Palma},\ and\ \citenamefont
  {Vacchini}}]{Lorenzo2017}%
  \BibitemOpen
  \bibfield  {author} {\bibinfo {author} {\bibfnamefont {S.}~\bibnamefont
  {Lorenzo}}, \bibinfo {author} {\bibfnamefont {F.}~\bibnamefont {Ciccarello}},
  \bibinfo {author} {\bibfnamefont {G.~M.}\ \bibnamefont {Palma}}, \ and\
  \bibinfo {author} {\bibfnamefont {B.}~\bibnamefont {Vacchini}},\ }\href
  {\doibase 10.1142/S123016121740011X} {\bibfield  {journal} {\bibinfo
  {journal} {Open Syst. Inf. Dyn.}\ }\textbf {\bibinfo {volume} {24}},\
  \bibinfo {pages} {1740011} (\bibinfo {year}
  {2017}{\natexlab{b}})}\BibitemShut {NoStop}%
\bibitem [{\citenamefont {Li}\ \emph {et~al.}(2018)\citenamefont {Li},
  \citenamefont {Hall},\ and\ \citenamefont {Wiseman}}]{Li2018}%
  \BibitemOpen
  \bibfield  {author} {\bibinfo {author} {\bibfnamefont {L.}~\bibnamefont
  {Li}}, \bibinfo {author} {\bibfnamefont {M.~J.}\ \bibnamefont {Hall}}, \ and\
  \bibinfo {author} {\bibfnamefont {H.~M.}\ \bibnamefont {Wiseman}},\ }\href
  {http://www.sciencedirect.com/science/article/pii/S0370157318301601}
  {\bibfield  {journal} {\bibinfo  {journal} {Phys. Rep.}\ }\textbf {\bibinfo
  {volume} {759}} (\bibinfo {year} {2018})}\BibitemShut {NoStop}%
\bibitem [{\citenamefont {Oreshkov}\ \emph {et~al.}(2012)\citenamefont
  {Oreshkov}, \citenamefont {Costa},\ and\ \citenamefont
  {Brukner}}]{Oreshkov2012}%
  \BibitemOpen
  \bibfield  {author} {\bibinfo {author} {\bibfnamefont {O.}~\bibnamefont
  {Oreshkov}}, \bibinfo {author} {\bibfnamefont {F.}~\bibnamefont {Costa}}, \
  and\ \bibinfo {author} {\bibfnamefont {{\v C}.}~\bibnamefont {Brukner}},\
  }\href {\doibase 10.1038/ncomms2076} {\bibfield  {journal} {\bibinfo
  {journal} {Nat. Commun.}\ }\textbf {\bibinfo {volume} {3}},\ \bibinfo {pages}
  {1092} (\bibinfo {year} {2012})}\BibitemShut {NoStop}%
\bibitem [{\citenamefont {Pollock}\ \emph
  {et~al.}(2018{\natexlab{a}})\citenamefont {Pollock}, \citenamefont
  {Rodr\'{\i}guez-Rosario}, \citenamefont {Frauenheim}, \citenamefont
  {Paternostro},\ and\ \citenamefont {Modi}}]{Pollock2018A}%
  \BibitemOpen
  \bibfield  {author} {\bibinfo {author} {\bibfnamefont {F.~A.}\ \bibnamefont
  {Pollock}}, \bibinfo {author} {\bibfnamefont {C.}~\bibnamefont
  {Rodr\'{\i}guez-Rosario}}, \bibinfo {author} {\bibfnamefont {T.}~\bibnamefont
  {Frauenheim}}, \bibinfo {author} {\bibfnamefont {M.}~\bibnamefont
  {Paternostro}}, \ and\ \bibinfo {author} {\bibfnamefont {K.}~\bibnamefont
  {Modi}},\ }\href {\doibase 10.1103/PhysRevA.97.012127} {\bibfield  {journal}
  {\bibinfo  {journal} {Phys. Rev. A}\ }\textbf {\bibinfo {volume} {97}},\
  \bibinfo {pages} {012127} (\bibinfo {year} {2018}{\natexlab{a}})}\BibitemShut
  {NoStop}%
\bibitem [{\citenamefont {Pollock}\ \emph
  {et~al.}(2018{\natexlab{b}})\citenamefont {Pollock}, \citenamefont
  {Rodr\'{\i}guez-Rosario}, \citenamefont {Frauenheim}, \citenamefont
  {Paternostro},\ and\ \citenamefont {Modi}}]{Pollock2018L}%
  \BibitemOpen
  \bibfield  {author} {\bibinfo {author} {\bibfnamefont {F.~A.}\ \bibnamefont
  {Pollock}}, \bibinfo {author} {\bibfnamefont {C.}~\bibnamefont
  {Rodr\'{\i}guez-Rosario}}, \bibinfo {author} {\bibfnamefont {T.}~\bibnamefont
  {Frauenheim}}, \bibinfo {author} {\bibfnamefont {M.}~\bibnamefont
  {Paternostro}}, \ and\ \bibinfo {author} {\bibfnamefont {K.}~\bibnamefont
  {Modi}},\ }\href {\doibase 10.1103/PhysRevLett.120.040405} {\bibfield
  {journal} {\bibinfo  {journal} {Phys. Rev. Lett.}\ }\textbf {\bibinfo
  {volume} {120}},\ \bibinfo {pages} {040405} (\bibinfo {year}
  {2018}{\natexlab{b}})}\BibitemShut {NoStop}%
\bibitem [{\citenamefont {Chiribella}\ \emph
  {et~al.}(2008{\natexlab{a}})\citenamefont {Chiribella}, \citenamefont
  {D'Ariano},\ and\ \citenamefont {Perinotti}}]{Chiribella2008}%
  \BibitemOpen
  \bibfield  {author} {\bibinfo {author} {\bibfnamefont {G.}~\bibnamefont
  {Chiribella}}, \bibinfo {author} {\bibfnamefont {G.~M.}\ \bibnamefont
  {D'Ariano}}, \ and\ \bibinfo {author} {\bibfnamefont {P.}~\bibnamefont
  {Perinotti}},\ }\href {\doibase 10.1209/0295-5075/83/30004} {\bibfield
  {journal} {\bibinfo  {journal} {Europhys. Lett.}\ }\textbf {\bibinfo {volume}
  {83}},\ \bibinfo {pages} {30004} (\bibinfo {year}
  {2008}{\natexlab{a}})}\BibitemShut {NoStop}%
\bibitem [{\citenamefont {Chiribella}\ \emph
  {et~al.}(2008{\natexlab{b}})\citenamefont {Chiribella}, \citenamefont
  {D'Ariano},\ and\ \citenamefont {Perinotti}}]{Chiribella2008-2}%
  \BibitemOpen
  \bibfield  {author} {\bibinfo {author} {\bibfnamefont {G.}~\bibnamefont
  {Chiribella}}, \bibinfo {author} {\bibfnamefont {G.~M.}\ \bibnamefont
  {D'Ariano}}, \ and\ \bibinfo {author} {\bibfnamefont {P.}~\bibnamefont
  {Perinotti}},\ }\href {\doibase 10.1103/PhysRevLett.101.060401} {\bibfield
  {journal} {\bibinfo  {journal} {Phys. Rev. Lett.}\ }\textbf {\bibinfo
  {volume} {101}},\ \bibinfo {pages} {060401} (\bibinfo {year}
  {2008}{\natexlab{b}})}\BibitemShut {NoStop}%
\bibitem [{\citenamefont {Chiribella}\ \emph {et~al.}(2009)\citenamefont
  {Chiribella}, \citenamefont {D'Ariano},\ and\ \citenamefont
  {Perinotti}}]{Chiribella2009}%
  \BibitemOpen
  \bibfield  {author} {\bibinfo {author} {\bibfnamefont {G.}~\bibnamefont
  {Chiribella}}, \bibinfo {author} {\bibfnamefont {G.~M.}\ \bibnamefont
  {D'Ariano}}, \ and\ \bibinfo {author} {\bibfnamefont {P.}~\bibnamefont
  {Perinotti}},\ }\href {\doibase 10.1103/PhysRevA.80.022339} {\bibfield
  {journal} {\bibinfo  {journal} {Phys. Rev. A}\ }\textbf {\bibinfo {volume}
  {80}},\ \bibinfo {pages} {022339} (\bibinfo {year} {2009})}\BibitemShut
  {NoStop}%
\bibitem [{\citenamefont {Shrapnel}\ \emph {et~al.}(2017)\citenamefont
  {Shrapnel}, \citenamefont {Costa},\ and\ \citenamefont
  {Milburn}}]{ShrapnelCosta2017}%
  \BibitemOpen
  \bibfield  {author} {\bibinfo {author} {\bibfnamefont {S.}~\bibnamefont
  {Shrapnel}}, \bibinfo {author} {\bibfnamefont {F.}~\bibnamefont {Costa}}, \
  and\ \bibinfo {author} {\bibfnamefont {G.}~\bibnamefont {Milburn}},\ }\href
  {https://arxiv.org/abs/1702.01845} {\bibfield  {journal} {\bibinfo  {journal}
  {arXiv:1702.01845}\ } (\bibinfo {year} {2017})}\BibitemShut {NoStop}%
\bibitem [{\citenamefont {Pearl}(2000)}]{Pearl}%
  \BibitemOpen
  \bibfield  {author} {\bibinfo {author} {\bibfnamefont {J.}~\bibnamefont
  {Pearl}},\ }\href@noop {} {\emph {\bibinfo {title} {{Causality}}}}\ (\bibinfo
   {publisher} {Oxford University Press, New York},\ \bibinfo {year}
  {2000})\BibitemShut {NoStop}%
\bibitem [{\citenamefont {Ara{\'{u}}jo}\ \emph {et~al.}(2015)\citenamefont
  {Ara{\'{u}}jo}, \citenamefont {Branciard}, \citenamefont {Costa},
  \citenamefont {Feix}, \citenamefont {Giarmatzi},\ and\ \citenamefont
  {Brukner}}]{Araujo2015}%
  \BibitemOpen
  \bibfield  {author} {\bibinfo {author} {\bibfnamefont {M.}~\bibnamefont
  {Ara{\'{u}}jo}}, \bibinfo {author} {\bibfnamefont {C.}~\bibnamefont
  {Branciard}}, \bibinfo {author} {\bibfnamefont {F.}~\bibnamefont {Costa}},
  \bibinfo {author} {\bibfnamefont {A.}~\bibnamefont {Feix}}, \bibinfo {author}
  {\bibfnamefont {C.}~\bibnamefont {Giarmatzi}}, \ and\ \bibinfo {author}
  {\bibfnamefont {{\v{C}}.}~\bibnamefont {Brukner}},\ }\href {\doibase
  10.1088/1367-2630/17/10/102001} {\bibfield  {journal} {\bibinfo  {journal}
  {New J. Phys.}\ }\textbf {\bibinfo {volume} {17}},\ \bibinfo {pages} {102001}
  (\bibinfo {year} {2015})}\BibitemShut {NoStop}%
\bibitem [{\citenamefont {Costa}\ and\ \citenamefont
  {Shrapnel}(2016)}]{Costa2016}%
  \BibitemOpen
  \bibfield  {author} {\bibinfo {author} {\bibfnamefont {F.}~\bibnamefont
  {Costa}}\ and\ \bibinfo {author} {\bibfnamefont {S.}~\bibnamefont
  {Shrapnel}},\ }\href {\doibase 10.1088/1367-2630/18/6/063032} {\bibfield
  {journal} {\bibinfo  {journal} {New J. Phys.}\ }\textbf {\bibinfo {volume}
  {18}},\ \bibinfo {pages} {063032} (\bibinfo {year} {2016})}\BibitemShut
  {NoStop}%
\bibitem [{\citenamefont {Oreshkov}\ and\ \citenamefont
  {Giarmatzi}(2016)}]{Oreshkov2016}%
  \BibitemOpen
  \bibfield  {author} {\bibinfo {author} {\bibfnamefont {O.}~\bibnamefont
  {Oreshkov}}\ and\ \bibinfo {author} {\bibfnamefont {C.}~\bibnamefont
  {Giarmatzi}},\ }\href {\doibase 10.1088/1367-2630/18/9/093020} {\bibfield
  {journal} {\bibinfo  {journal} {New J. Phys.}\ }\textbf {\bibinfo {volume}
  {18}},\ \bibinfo {pages} {093020} (\bibinfo {year} {2016})}\BibitemShut
  {NoStop}%
\bibitem [{\citenamefont {Allen}\ \emph {et~al.}(2017)\citenamefont {Allen},
  \citenamefont {Barrett}, \citenamefont {Horsman}, \citenamefont {Lee},\ and\
  \citenamefont {Spekkens}}]{Allen2017}%
  \BibitemOpen
  \bibfield  {author} {\bibinfo {author} {\bibfnamefont {J.-M.~A.}\
  \bibnamefont {Allen}}, \bibinfo {author} {\bibfnamefont {J.}~\bibnamefont
  {Barrett}}, \bibinfo {author} {\bibfnamefont {D.~C.}\ \bibnamefont
  {Horsman}}, \bibinfo {author} {\bibfnamefont {C.~M.}\ \bibnamefont {Lee}}, \
  and\ \bibinfo {author} {\bibfnamefont {R.~W.}\ \bibnamefont {Spekkens}},\
  }\href {\doibase 10.1103/PhysRevX.7.031021} {\bibfield  {journal} {\bibinfo
  {journal} {Phys. Rev. X}\ }\textbf {\bibinfo {volume} {7}},\ \bibinfo {pages}
  {031021} (\bibinfo {year} {2017})}\BibitemShut {NoStop}%
\bibitem [{\citenamefont {Ringbauer}\ and\ \citenamefont
  {Chaves}(2017)}]{Ringbauer2017}%
  \BibitemOpen
  \bibfield  {author} {\bibinfo {author} {\bibfnamefont {M.}~\bibnamefont
  {Ringbauer}}\ and\ \bibinfo {author} {\bibfnamefont {R.}~\bibnamefont
  {Chaves}},\ }\href {\doibase 10.22331/q-2017-11-25-35} {\bibfield  {journal}
  {\bibinfo  {journal} {Quantum}\ }\textbf {\bibinfo {volume} {1}},\ \bibinfo
  {pages} {35} (\bibinfo {year} {2017})}\BibitemShut {NoStop}%
\bibitem [{\citenamefont {Taranto}\ \emph {et~al.}(2019)\citenamefont
  {Taranto}, \citenamefont {Milz}, \citenamefont {Pollock},\ and\ \citenamefont
  {Modi}}]{Taranto2018A}%
  \BibitemOpen
  \bibfield  {author} {\bibinfo {author} {\bibfnamefont {P.}~\bibnamefont
  {Taranto}}, \bibinfo {author} {\bibfnamefont {S.}~\bibnamefont {Milz}},
  \bibinfo {author} {\bibfnamefont {F.~A.}\ \bibnamefont {Pollock}}, \ and\
  \bibinfo {author} {\bibfnamefont {K.}~\bibnamefont {Modi}},\ }\href {\doibase
  10.1103/PhysRevA.99.042108} {\bibfield  {journal} {\bibinfo  {journal} {Phys.
  Rev. A}\ }\textbf {\bibinfo {volume} {99}},\ \bibinfo {pages} {042108}
  (\bibinfo {year} {2019})}\BibitemShut {NoStop}%
\bibitem [{Note1()}]{Note1}%
  \BibitemOpen
  \bibinfo {note} {These are themselves often referred to as memory. Here, we
  distinguish the total temporal correlations between observables from those
  resulting from \protect \textit {non-Markovian} memory, which can survive
  interventions that reset the system's state.}\BibitemShut {Stop}%
\bibitem [{\citenamefont {Ruskai}(2002)}]{Ruskai2002}%
  \BibitemOpen
  \bibfield  {author} {\bibinfo {author} {\bibfnamefont {M.~B.}\ \bibnamefont
  {Ruskai}},\ }\href {\doibase 10.1063/1.1497701} {\bibfield  {journal}
  {\bibinfo  {journal} {J. Math. Phys.}\ }\textbf {\bibinfo {volume} {43}},\
  \bibinfo {pages} {4358} (\bibinfo {year} {2002})}\BibitemShut {NoStop}%
\bibitem [{\citenamefont {Petz}(2003)}]{Petz2003}%
  \BibitemOpen
  \bibfield  {author} {\bibinfo {author} {\bibfnamefont {D.}~\bibnamefont
  {Petz}},\ }\href {\doibase 10.1142/S0129055X03001576} {\bibfield  {journal}
  {\bibinfo  {journal} {Rev. Math. Phys.}\ }\textbf {\bibinfo {volume} {15}},\
  \bibinfo {pages} {79} (\bibinfo {year} {2003})}\BibitemShut {NoStop}%
\bibitem [{\citenamefont {Hayden}\ \emph {et~al.}(2004)\citenamefont {Hayden},
  \citenamefont {Jozsa}, \citenamefont {Petz},\ and\ \citenamefont
  {Winter}}]{Hayden2004}%
  \BibitemOpen
  \bibfield  {author} {\bibinfo {author} {\bibfnamefont {P.}~\bibnamefont
  {Hayden}}, \bibinfo {author} {\bibfnamefont {R.}~\bibnamefont {Jozsa}},
  \bibinfo {author} {\bibfnamefont {D.}~\bibnamefont {Petz}}, \ and\ \bibinfo
  {author} {\bibfnamefont {A.}~\bibnamefont {Winter}},\ }\href {\doibase
  10.1007/s00220-004-1049-z} {\bibfield  {journal} {\bibinfo  {journal}
  {Commun. Math. Phys.}\ }\textbf {\bibinfo {volume} {246}},\ \bibinfo {pages}
  {359} (\bibinfo {year} {2004})}\BibitemShut {NoStop}%
\bibitem [{\citenamefont {Ibinson}\ \emph {et~al.}(2008)\citenamefont
  {Ibinson}, \citenamefont {Linden},\ and\ \citenamefont
  {Winter}}]{Ibinson2008}%
  \BibitemOpen
  \bibfield  {author} {\bibinfo {author} {\bibfnamefont {B.}~\bibnamefont
  {Ibinson}}, \bibinfo {author} {\bibfnamefont {N.}~\bibnamefont {Linden}}, \
  and\ \bibinfo {author} {\bibfnamefont {A.}~\bibnamefont {Winter}},\ }\href
  {\doibase 10.1007/s00220-007-0362-8} {\bibfield  {journal} {\bibinfo
  {journal} {Commun. Math. Phys.}\ }\textbf {\bibinfo {volume} {277}},\
  \bibinfo {pages} {289} (\bibinfo {year} {2008})}\BibitemShut {NoStop}%
\bibitem [{\citenamefont {Fawzi}\ and\ \citenamefont
  {Renner}(2015)}]{Fawzi2015}%
  \BibitemOpen
  \bibfield  {author} {\bibinfo {author} {\bibfnamefont {O.}~\bibnamefont
  {Fawzi}}\ and\ \bibinfo {author} {\bibfnamefont {R.}~\bibnamefont {Renner}},\
  }\href {\doibase 10.1007/s00220-015-2466-x} {\bibfield  {journal} {\bibinfo
  {journal} {Commun. Math. Phys.}\ }\textbf {\bibinfo {volume} {340}},\
  \bibinfo {pages} {575} (\bibinfo {year} {2015})}\BibitemShut {NoStop}%
\bibitem [{\citenamefont {Wilde}(2015)}]{Wilde2015}%
  \BibitemOpen
  \bibfield  {author} {\bibinfo {author} {\bibfnamefont {M.~M.}\ \bibnamefont
  {Wilde}},\ }\href {\doibase 10.1098/rspa.2015.0338} {\bibfield  {journal}
  {\bibinfo  {journal} {Proc. Royal Soc. A}\ }\textbf {\bibinfo {volume}
  {471}},\ \bibinfo {pages} {2182} (\bibinfo {year} {2015})}\BibitemShut
  {NoStop}%
\bibitem [{\citenamefont {Sutter}\ \emph {et~al.}(2016)\citenamefont {Sutter},
  \citenamefont {Fawzi},\ and\ \citenamefont {Renner}}]{Sutter2016}%
  \BibitemOpen
  \bibfield  {author} {\bibinfo {author} {\bibfnamefont {D.}~\bibnamefont
  {Sutter}}, \bibinfo {author} {\bibfnamefont {O.}~\bibnamefont {Fawzi}}, \
  and\ \bibinfo {author} {\bibfnamefont {R.}~\bibnamefont {Renner}},\ }\href
  {\doibase 10.1098/rspa.2015.0623} {\bibfield  {journal} {\bibinfo  {journal}
  {Proc. Royal Soc. A}\ }\textbf {\bibinfo {volume} {472}},\ \bibinfo {pages}
  {2186} (\bibinfo {year} {2016})}\BibitemShut {NoStop}%
\bibitem [{\citenamefont {Sutter}\ \emph {et~al.}(2017)\citenamefont {Sutter},
  \citenamefont {Berta},\ and\ \citenamefont {Tomamichel}}]{Sutter2017}%
  \BibitemOpen
  \bibfield  {author} {\bibinfo {author} {\bibfnamefont {D.}~\bibnamefont
  {Sutter}}, \bibinfo {author} {\bibfnamefont {M.}~\bibnamefont {Berta}}, \
  and\ \bibinfo {author} {\bibfnamefont {M.}~\bibnamefont {Tomamichel}},\
  }\href {\doibase 10.1007/s00220-016-2778-5} {\bibfield  {journal} {\bibinfo
  {journal} {Commun. Math. Phys.}\ }\textbf {\bibinfo {volume} {352}},\
  \bibinfo {pages} {37} (\bibinfo {year} {2017})}\BibitemShut {NoStop}%
\bibitem [{\citenamefont {Jamio{\l}kowski}(1972)}]{Jamiolkowski1972}%
  \BibitemOpen
  \bibfield  {author} {\bibinfo {author} {\bibfnamefont {A.}~\bibnamefont
  {Jamio{\l}kowski}},\ }\href {\doibase
  https://doi.org/10.1016/0034-4877(72)90011-0} {\bibfield  {journal} {\bibinfo
   {journal} {Rep. Math. Phys.}\ }\textbf {\bibinfo {volume} {3}},\ \bibinfo
  {pages} {275} (\bibinfo {year} {1972})}\BibitemShut {NoStop}%
\bibitem [{\citenamefont {Choi}(1975)}]{Choi1975}%
  \BibitemOpen
  \bibfield  {author} {\bibinfo {author} {\bibfnamefont {M.-D.}\ \bibnamefont
  {Choi}},\ }\href {\doibase https://doi.org/10.1016/0024-3795(75)90075-0}
  {\bibfield  {journal} {\bibinfo  {journal} {Linear Algebra Appl.}\ }\textbf
  {\bibinfo {volume} {10}},\ \bibinfo {pages} {285} (\bibinfo {year}
  {1975})}\BibitemShut {NoStop}%
\bibitem [{\citenamefont {Milz}\ \emph {et~al.}(2018)\citenamefont {Milz},
  \citenamefont {Pollock}, \citenamefont {Le}, \citenamefont {Chiribella},\
  and\ \citenamefont {Modi}}]{Milz2018}%
  \BibitemOpen
  \bibfield  {author} {\bibinfo {author} {\bibfnamefont {S.}~\bibnamefont
  {Milz}}, \bibinfo {author} {\bibfnamefont {F.~A.}\ \bibnamefont {Pollock}},
  \bibinfo {author} {\bibfnamefont {T.~P.}\ \bibnamefont {Le}}, \bibinfo
  {author} {\bibfnamefont {G.}~\bibnamefont {Chiribella}}, \ and\ \bibinfo
  {author} {\bibfnamefont {K.}~\bibnamefont {Modi}},\ }\href {\doibase
  10.1088/1367-2630/aaafee} {\bibfield  {journal} {\bibinfo  {journal} {New J.
  Phys.}\ }\textbf {\bibinfo {volume} {20}},\ \bibinfo {pages} {033033}
  (\bibinfo {year} {2018})}\BibitemShut {NoStop}%
\bibitem [{\citenamefont {Ringbauer}\ \emph {et~al.}(2015)\citenamefont
  {Ringbauer}, \citenamefont {Wood}, \citenamefont {Modi}, \citenamefont
  {Gilchrist}, \citenamefont {White},\ and\ \citenamefont
  {Fedrizzi}}]{Ringbauer2015}%
  \BibitemOpen
  \bibfield  {author} {\bibinfo {author} {\bibfnamefont {M.}~\bibnamefont
  {Ringbauer}}, \bibinfo {author} {\bibfnamefont {C.~J.}\ \bibnamefont {Wood}},
  \bibinfo {author} {\bibfnamefont {K.}~\bibnamefont {Modi}}, \bibinfo {author}
  {\bibfnamefont {A.}~\bibnamefont {Gilchrist}}, \bibinfo {author}
  {\bibfnamefont {A.~G.}\ \bibnamefont {White}}, \ and\ \bibinfo {author}
  {\bibfnamefont {A.}~\bibnamefont {Fedrizzi}},\ }\href {\doibase
  10.1103/PhysRevLett.114.090402} {\bibfield  {journal} {\bibinfo  {journal}
  {Phys. Rev. Lett.}\ }\textbf {\bibinfo {volume} {114}},\ \bibinfo {pages}
  {090402} (\bibinfo {year} {2015})}\BibitemShut {NoStop}%
\bibitem [{\citenamefont {Kleinhans}\ \emph {et~al.}(2007)\citenamefont
  {Kleinhans}, \citenamefont {Friedrich}, \citenamefont {W\"achter},\ and\
  \citenamefont {Peinke}}]{Kleinhans2007}%
  \BibitemOpen
  \bibfield  {author} {\bibinfo {author} {\bibfnamefont {D.}~\bibnamefont
  {Kleinhans}}, \bibinfo {author} {\bibfnamefont {R.}~\bibnamefont
  {Friedrich}}, \bibinfo {author} {\bibfnamefont {M.}~\bibnamefont
  {W\"achter}}, \ and\ \bibinfo {author} {\bibfnamefont {J.}~\bibnamefont
  {Peinke}},\ }\href {\doibase 10.1103/PhysRevE.76.041109} {\bibfield
  {journal} {\bibinfo  {journal} {Phys. Rev. E}\ }\textbf {\bibinfo {volume}
  {76}},\ \bibinfo {pages} {041109} (\bibinfo {year} {2007})}\BibitemShut
  {NoStop}%
\bibitem [{\citenamefont {Wiesner}\ and\ \citenamefont
  {Crutchfield}(2006)}]{Wiesner2006}%
  \BibitemOpen
  \bibfield  {author} {\bibinfo {author} {\bibfnamefont {K.}~\bibnamefont
  {Wiesner}}\ and\ \bibinfo {author} {\bibfnamefont {J.~P.}\ \bibnamefont
  {Crutchfield}},\ }\href {https://arxiv.org/abs/quant-ph/0611143} {\bibfield
  {journal} {\bibinfo  {journal} {arXiv:0611143}\ } (\bibinfo {year}
  {2006})}\BibitemShut {NoStop}%
\bibitem [{\citenamefont {Siefert}\ \emph {et~al.}(2003)\citenamefont
  {Siefert}, \citenamefont {Kittel}, \citenamefont {Friedrich},\ and\
  \citenamefont {Peinke}}]{Siefert2003}%
  \BibitemOpen
  \bibfield  {author} {\bibinfo {author} {\bibfnamefont {M.}~\bibnamefont
  {Siefert}}, \bibinfo {author} {\bibfnamefont {A.}~\bibnamefont {Kittel}},
  \bibinfo {author} {\bibfnamefont {R.}~\bibnamefont {Friedrich}}, \ and\
  \bibinfo {author} {\bibfnamefont {J.}~\bibnamefont {Peinke}},\ }\href
  {\doibase 10.1209/epl/i2003-00152-9} {\bibfield  {journal} {\bibinfo
  {journal} {Europhys. Lett.}\ }\textbf {\bibinfo {volume} {61}},\ \bibinfo
  {pages} {466} (\bibinfo {year} {2003})}\BibitemShut {NoStop}%
\bibitem [{\citenamefont {B\"ottcher}\ \emph {et~al.}(2006)\citenamefont
  {B\"ottcher}, \citenamefont {Peinke}, \citenamefont {Kleinhans},
  \citenamefont {Friedrich}, \citenamefont {Lind},\ and\ \citenamefont
  {Haase}}]{Bottcher2006}%
  \BibitemOpen
  \bibfield  {author} {\bibinfo {author} {\bibfnamefont {F.}~\bibnamefont
  {B\"ottcher}}, \bibinfo {author} {\bibfnamefont {J.}~\bibnamefont {Peinke}},
  \bibinfo {author} {\bibfnamefont {D.}~\bibnamefont {Kleinhans}}, \bibinfo
  {author} {\bibfnamefont {R.}~\bibnamefont {Friedrich}}, \bibinfo {author}
  {\bibfnamefont {P.~G.}\ \bibnamefont {Lind}}, \ and\ \bibinfo {author}
  {\bibfnamefont {M.}~\bibnamefont {Haase}},\ }\href {\doibase
  10.1103/PhysRevLett.97.090603} {\bibfield  {journal} {\bibinfo  {journal}
  {Phys. Rev. Lett.}\ }\textbf {\bibinfo {volume} {97}},\ \bibinfo {pages}
  {090603} (\bibinfo {year} {2006})}\BibitemShut {NoStop}%
\bibitem [{\citenamefont {Lehle}(2011)}]{Lehle2011}%
  \BibitemOpen
  \bibfield  {author} {\bibinfo {author} {\bibfnamefont {B.}~\bibnamefont
  {Lehle}},\ }\href {\doibase 10.1103/PhysRevE.83.021113} {\bibfield  {journal}
  {\bibinfo  {journal} {Phys. Rev. E}\ }\textbf {\bibinfo {volume} {83}},\
  \bibinfo {pages} {021113} (\bibinfo {year} {2011})}\BibitemShut {NoStop}%
\bibitem [{\citenamefont {Makri}\ and\ \citenamefont
  {Makarov}(1995{\natexlab{a}})}]{Makri1995JCP-1}%
  \BibitemOpen
  \bibfield  {author} {\bibinfo {author} {\bibfnamefont {N.}~\bibnamefont
  {Makri}}\ and\ \bibinfo {author} {\bibfnamefont {D.~E.}\ \bibnamefont
  {Makarov}},\ }\href {\doibase 10.1063/1.469508} {\bibfield  {journal}
  {\bibinfo  {journal} {J. Chem. Phys.}\ }\textbf {\bibinfo {volume} {102}},\
  \bibinfo {pages} {4600} (\bibinfo {year} {1995}{\natexlab{a}})}\BibitemShut
  {NoStop}%
\bibitem [{\citenamefont {Makri}\ and\ \citenamefont
  {Makarov}(1995{\natexlab{b}})}]{Makri1995JCP-2}%
  \BibitemOpen
  \bibfield  {author} {\bibinfo {author} {\bibfnamefont {N.}~\bibnamefont
  {Makri}}\ and\ \bibinfo {author} {\bibfnamefont {D.~E.}\ \bibnamefont
  {Makarov}},\ }\href {\doibase 10.1063/1.469509} {\bibfield  {journal}
  {\bibinfo  {journal} {J. Chem. Phys.}\ }\textbf {\bibinfo {volume} {102}},\
  \bibinfo {pages} {4611} (\bibinfo {year} {1995}{\natexlab{b}})}\BibitemShut
  {NoStop}%
\bibitem [{\citenamefont {de~Vega}\ and\ \citenamefont
  {Alonso}(2017)}]{deVega2017}%
  \BibitemOpen
  \bibfield  {author} {\bibinfo {author} {\bibfnamefont {I.}~\bibnamefont
  {de~Vega}}\ and\ \bibinfo {author} {\bibfnamefont {D.}~\bibnamefont
  {Alonso}},\ }\href {\doibase 10.1103/RevModPhys.89.015001} {\bibfield
  {journal} {\bibinfo  {journal} {Rev. Mod. Phys.}\ }\textbf {\bibinfo {volume}
  {89}},\ \bibinfo {pages} {015001} (\bibinfo {year} {2017})}\BibitemShut
  {NoStop}%
\bibitem [{\citenamefont {Pollock}\ and\ \citenamefont
  {Modi}(2018{\natexlab{b}})}]{Pollock2018Q}%
  \BibitemOpen
  \bibfield  {author} {\bibinfo {author} {\bibfnamefont {F.~A.}\ \bibnamefont
  {Pollock}}\ and\ \bibinfo {author} {\bibfnamefont {K.}~\bibnamefont {Modi}},\
  }\href {\doibase 10.22331/q-2018-07-11-76} {\bibfield  {journal} {\bibinfo
  {journal} {{Quantum}}\ }\textbf {\bibinfo {volume} {2}},\ \bibinfo {pages}
  {76} (\bibinfo {year} {2018}{\natexlab{b}})}\BibitemShut {NoStop}%
\bibitem [{\citenamefont {Clemens}\ \emph {et~al.}(2004)\citenamefont
  {Clemens}, \citenamefont {Siddiqui},\ and\ \citenamefont
  {Gea-Banacloche}}]{Clemens2004}%
  \BibitemOpen
  \bibfield  {author} {\bibinfo {author} {\bibfnamefont {J.~P.}\ \bibnamefont
  {Clemens}}, \bibinfo {author} {\bibfnamefont {S.}~\bibnamefont {Siddiqui}}, \
  and\ \bibinfo {author} {\bibfnamefont {J.}~\bibnamefont {Gea-Banacloche}},\
  }\href {\doibase 10.1103/PhysRevA.69.062313} {\bibfield  {journal} {\bibinfo
  {journal} {Phys. Rev. A}\ }\textbf {\bibinfo {volume} {69}},\ \bibinfo
  {pages} {062313} (\bibinfo {year} {2004})}\BibitemShut {NoStop}%
\bibitem [{\citenamefont {Chiribella}\ \emph {et~al.}(2011)\citenamefont
  {Chiribella}, \citenamefont {Dall'Arno}, \citenamefont {D'Ariano},
  \citenamefont {Macchiavello},\ and\ \citenamefont
  {Perinotti}}]{Chiribella2011}%
  \BibitemOpen
  \bibfield  {author} {\bibinfo {author} {\bibfnamefont {G.}~\bibnamefont
  {Chiribella}}, \bibinfo {author} {\bibfnamefont {M.}~\bibnamefont
  {Dall'Arno}}, \bibinfo {author} {\bibfnamefont {G.~M.}\ \bibnamefont
  {D'Ariano}}, \bibinfo {author} {\bibfnamefont {C.}~\bibnamefont
  {Macchiavello}}, \ and\ \bibinfo {author} {\bibfnamefont {P.}~\bibnamefont
  {Perinotti}},\ }\href {\doibase 10.1103/PhysRevA.83.052305} {\bibfield
  {journal} {\bibinfo  {journal} {Phys. Rev. A}\ }\textbf {\bibinfo {volume}
  {83}},\ \bibinfo {pages} {052305} (\bibinfo {year} {2011})}\BibitemShut
  {NoStop}%
\bibitem [{\citenamefont {Ben-Aroya}\ and\ \citenamefont
  {Ta-Shma}(2011)}]{Ben-Aroya2011}%
  \BibitemOpen
  \bibfield  {author} {\bibinfo {author} {\bibfnamefont {A.}~\bibnamefont
  {Ben-Aroya}}\ and\ \bibinfo {author} {\bibfnamefont {A.}~\bibnamefont
  {Ta-Shma}},\ }\href {\doibase 10.1109/TIT.2011.2134410} {\bibfield  {journal}
  {\bibinfo  {journal} {IEEE Trans. Inf. Theory}\ }\textbf {\bibinfo {volume}
  {57}},\ \bibinfo {pages} {3982} (\bibinfo {year} {2011})}\BibitemShut
  {NoStop}%
\bibitem [{\citenamefont {Lupo}\ \emph {et~al.}(2012)\citenamefont {Lupo},
  \citenamefont {Memarzadeh},\ and\ \citenamefont {Mancini}}]{Lupo2012}%
  \BibitemOpen
  \bibfield  {author} {\bibinfo {author} {\bibfnamefont {C.}~\bibnamefont
  {Lupo}}, \bibinfo {author} {\bibfnamefont {L.}~\bibnamefont {Memarzadeh}}, \
  and\ \bibinfo {author} {\bibfnamefont {S.}~\bibnamefont {Mancini}},\ }\href
  {\doibase 10.1103/PhysRevA.85.012320} {\bibfield  {journal} {\bibinfo
  {journal} {Phys. Rev. A}\ }\textbf {\bibinfo {volume} {85}},\ \bibinfo
  {pages} {012320} (\bibinfo {year} {2012})}\BibitemShut {NoStop}%
\bibitem [{\citenamefont {Edmunds}\ \emph {et~al.}(2017)\citenamefont
  {Edmunds}, \citenamefont {Hempel}, \citenamefont {Harris}, \citenamefont
  {Ball}, \citenamefont {Frey}, \citenamefont {Stace},\ and\ \citenamefont
  {Biercuk}}]{Edmunds2017}%
  \BibitemOpen
  \bibfield  {author} {\bibinfo {author} {\bibfnamefont {C.~L.}\ \bibnamefont
  {Edmunds}}, \bibinfo {author} {\bibfnamefont {C.}~\bibnamefont {Hempel}},
  \bibinfo {author} {\bibfnamefont {R.}~\bibnamefont {Harris}}, \bibinfo
  {author} {\bibfnamefont {H.}~\bibnamefont {Ball}}, \bibinfo {author}
  {\bibfnamefont {V.}~\bibnamefont {Frey}}, \bibinfo {author} {\bibfnamefont
  {T.~M.}\ \bibnamefont {Stace}}, \ and\ \bibinfo {author} {\bibfnamefont
  {M.~J.}\ \bibnamefont {Biercuk}},\ }\href {https://arxiv.org/abs/1712.04954}
  {\bibfield  {journal} {\bibinfo  {journal} {arXiv:1712.04954}\ } (\bibinfo
  {year} {2017})}\BibitemShut {NoStop}%
\bibitem [{\citenamefont {Grimsmo}(2015)}]{Grimsmo2015}%
  \BibitemOpen
  \bibfield  {author} {\bibinfo {author} {\bibfnamefont {A.~L.}\ \bibnamefont
  {Grimsmo}},\ }\href {\doibase 10.1103/PhysRevLett.115.060402} {\bibfield
  {journal} {\bibinfo  {journal} {Phys. Rev. Lett.}\ }\textbf {\bibinfo
  {volume} {115}},\ \bibinfo {pages} {060402} (\bibinfo {year}
  {2015})}\BibitemShut {NoStop}%
\bibitem [{\citenamefont {Whalen}\ \emph {et~al.}(2017)\citenamefont {Whalen},
  \citenamefont {Grimsmo},\ and\ \citenamefont {Carmichael}}]{Whalen2017}%
  \BibitemOpen
  \bibfield  {author} {\bibinfo {author} {\bibfnamefont {S.~J.}\ \bibnamefont
  {Whalen}}, \bibinfo {author} {\bibfnamefont {A.~L.}\ \bibnamefont {Grimsmo}},
  \ and\ \bibinfo {author} {\bibfnamefont {H.~J.}\ \bibnamefont {Carmichael}},\
  }\href {\doibase 10.1088/2058-9565/aa8331} {\bibfield  {journal} {\bibinfo
  {journal} {Quantum Sci. Technol.}\ }\textbf {\bibinfo {volume} {2}},\
  \bibinfo {pages} {044008} (\bibinfo {year} {2017})}\BibitemShut {NoStop}%
\bibitem [{Note2()}]{Note2}%
  \BibitemOpen
  \bibinfo {note} {In the original formulation, the process tensor also outputs
  the final density operator of the system~\cite {Pollock2018A}, which is not
  important to our analysis here. Further, note that Eq.~\protect \textup
  {\hbox {\mathsurround \z@ \protect \normalfont (\ignorespaces \ref
  {eq:generalizedborn}\unskip \@@italiccorr )}} (Eq.~\protect \textup {\hbox
  {\mathsurround \z@ \protect \normalfont (\ignorespaces \ref
  {eq:processtensor}\unskip \@@italiccorr )}} in the main text) is correct up
  to a transpose of $O_{n:0}^{(x_{n:0})}$; we do not explicitly indicate this
  transposition to ease notation, with no affect on any results.}\BibitemShut
  {Stop}%
\end{thebibliography}


%


\end{document}